\documentclass{article}

\usepackage[affil-it]{authblk}
\usepackage[usenames,dvipsnames]{xcolor}
\usepackage{amsfonts}
\usepackage{amsmath,amsthm,amssymb,dsfont}
\usepackage{enumerate}
\usepackage{graphicx}	
\usepackage{subcaption}
\usepackage[margin=3cm]{geometry}
\usepackage{url}
\usepackage{todonotes}
\usepackage{bbm}
\usepackage{xcolor}
\definecolor{linkblue}{HTML}{001487}
\definecolor{ibmblue}{HTML}{FF00FF}
\usepackage{hyperref}
\hypersetup{colorlinks=true,citecolor=linkblue,linkcolor=linkblue,filecolor=linkblue,urlcolor=linkblue,breaklinks=true}

\usepackage{boldline}

\usepackage{tikz}
\usetikzlibrary{chains}
\usetikzlibrary{fit}
\usepackage{pgfplots}
\pgfplotsset{compat=1.10}
\usepgfplotslibrary{fillbetween}

\usepackage{pifont}
\usepackage{multirow}

\usepackage{epsfig}
\usetikzlibrary{shapes.symbols,patterns} % for source symbols
\usepackage{pgfplots}
\usetikzlibrary{decorations.pathmorphing}
\usepackage{enumerate}

\usepackage{mathtools}
\usepackage[nameinlink,capitalize,noabbrev]{cleveref}

\usepackage{cite}

\usepackage{mathrsfs}

\newtheorem{theorem}{Theorem}[section]
\newtheorem*{theorem*}{Theorem}

\newtheorem{lemma}[theorem]{Lemma}

\newtheorem{corollary}[theorem]{Corollary}

\theoremstyle{remark}

\theoremstyle{definition}
\newtheorem{definition}[theorem]{Definition}

\numberwithin{equation}{section}

\newcommand{\ket}[1]{|#1\rangle}
\newcommand{\bra}[1]{\langle#1|}
\newcommand{\proj}[1]{\ket{#1}\!\bra{#1}}
\newcommand{\tr}[1]{\mbox{\rm tr}\!\left[ #1 \right]}

\newcommand*{\TPCP}{\mathrm{CPTP}}
\newcommand*{\TNCP}{\mathrm{CPTN}}
\newcommand*{\LO}{\mathrm{LO}}

\newcommand*{\LOCC}{\mathrm{LOCC}}
\newcommand*{\LOCCOW}{\mathrm{LO}\overrightarrow{\mathrm{CC}}}
\newcommand*{\CNOT}{\mathrm{CNOT}}
\newcommand*{\identity}{\mathrm{I}}

\newcommand*{\SWAP}{\mathrm{SWAP}}
\newcommand*{\iSWAP}{\mathrm{iSWAP}}
\newcommand*{\SEP}{\mathrm{SEP}}
\newcommand*{\ci}{\mathrm{i}}

\newcommand*{\id}{\mathds{1}}

\newcommand*{\ee}{\mathrm{e}}

\newcommand*{\cA}{\mathcal{A}}
\newcommand*{\cO}{O}
\newcommand*{\cB}{\mathcal{B}}

\newcommand*{\N}{\mathbb{N}}

\newcommand*{\R}{\mathbb{R}}

\allowdisplaybreaks

\title{Circuit knitting with classical communication}
 \author{\normalsize Christophe Piveteau$^{1}$\thanks{cpivetea@phys.ethz.ch} \, and David Sutter$^{2}$}
  \affil{\small $^{1}$Institute for Theoretical Physics, ETH Zurich, Switzerland\\
  \small $^{2}$IBM Quantum, IBM Research Europe -- Zurich, Switzerland
  }
 \date{}

\begin{document}

\maketitle
\begin{abstract}
The scarcity of qubits is a major obstacle to the practical usage of quantum computers in the near future. To circumvent this problem, various circuit knitting techniques have been developed to partition large quantum circuits into subcircuits that fit on smaller devices, at the cost of a simulation overhead. In this work, we study a particular method of circuit knitting based on quasiprobability simulation of nonlocal gates with operations that act locally on the subcircuits. We investigate whether classical communication between these local quantum computers can help. We provide a positive answer by showing that for circuits containing $n$ nonlocal $\CNOT$ gates connecting two circuit parts, the simulation overhead can be reduced from $\cO(9^n)$ to $\cO(4^n)$ if one allows for classical information exchange. Similar improvements can be obtained for general Clifford gates and, at least in a restricted form, for other gates such as controlled rotation gates.
\end{abstract}

\section{Introduction} \label{sec_intro}
One of the major challenges of near-term quantum computation is the limited number of available qubits.
The problem will remain pronounced in the early days of quantum error correction, since a considerable amount of physical qubits will likely be required to realize a single logical qubit.
This motivated a significant amount of research on techniques, sometimes called \emph{circuit knitting}, that allow us to simulate a quantum computer with more qubits than physically available~\cite{BSS16,PHOW20,Mitarai_2021,forging22}.
Circuit knitting could be of central importance for the first practical demonstrations of a quantum advantage for a useful task.
One possibility to realize circuit knitting is using the technique of \emph{quasiprobability simulation}, which has previously gained much interest in the fields of quantum error mitigation~\cite{TBG17,endo18,kandala19,PSBGT21,PSW22} and classical simulation algorithms~\cite{PWB15, HC17, SC19, HG19, SBHYC21}.

For an arbitrary quantum circuit we can group all the qubits into two separate regions $\bar A$ and $\bar B$, such that ideally there are only few gates acting on both regions at the same time, as seen on the left-hand side of~\cref{fig_cutting}.\footnote{Here we consider cutting the circuit into two parts for simplicity, the technique can be straightforwardly generalized to involve more cuts.}
The method of quasiprobability simulation allows us to obtain the expected value of the measurement outcomes of this circuit by only sampling outcomes from circuits where the nonlocal gates are probabilistically replaced by local operations, as seen on the right-hand side of~\cref{fig_cutting}.
That means, instead of having to use one large quantum computer to simulate the full circuit, the outcome of the circuit can be estimated with two smaller quantum computers acting only on $\bar A$ and $\bar B$.
The cost of this technique is a sampling overhead that scales exponentially in the number of nonlocal gates involved in the circuit.\footnote{The quasiprobability simulation constructs an unbiased estimator for the measurement outcome of the nonlocal circuit. This means it preserves the correct expectation value but the variance is increased which gives rise to the sampling overhead~\cite{TBG17,endo18,Piv_masterThesis}.}
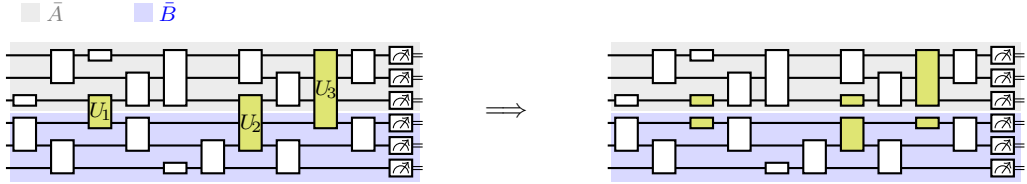
\begin{figure}[htb!]
    \centering
    \begin{tikzpicture}[thick, scale=1]
    \def\xs{0.07}
    \def \b{0.3}
    \def \x{8}
    \def \s{0.02}

     \draw [fill=gray!15,draw=none] (-0.05,0.75+2.5*\xs) rectangle (5.4,0+0.25*\xs);   
     \draw [fill=blue!15,draw=none] (-0.05,-0.25*\xs) rectangle (5.4,-0.75-2.5*\xs);       
    
    \draw (-0.1,0.45) -- (5,0.45);   
    \draw (-0.1,0.75) -- (5,0.75);  
    \draw (-0.1,0.15) -- (5,0.15);
    \draw (-0.1,-0.15) -- (5,-0.15);
    \draw (-0.1,-0.45) -- (5,-0.45);   
    \draw (-0.1,-0.75) -- (5,-0.75);   
    \draw [fill=GreenYellow,draw=black] (1,0.15+\xs) rectangle (1+\b,-0.15-\xs); 
    \draw [fill=GreenYellow,draw=black] (3,0.15+\xs) rectangle (3+\b,-0.45-\xs); 
    \draw [fill=GreenYellow,draw=black] (4,0.75+\xs) rectangle (4+\b,-0.15-\xs); 
    \node at (1+\b/2+0.01,0) {\footnotesize{$U_{\!1}$}};
    \node at (3+\b/2+0.01,-0.15) {\footnotesize{$U_{\!2}$}};
    \node at (4+\b/2+0.01,0.3) {\footnotesize{$U_{\!3}$}};    
    
    \draw [fill=white,draw=black] (0.5,0.75+\xs) rectangle (0.5+\b,0.45-\xs);
    \draw [fill=white,draw=black] (0.0,0.15+\xs) rectangle (0.0+\b,0.15-\xs);   
    \draw [fill=white,draw=black] (0.0,-0.15+\xs) rectangle (0.0+\b,-0.45-\xs);
    \draw [fill=white,draw=black] (0.5,-0.45+\xs) rectangle (0.5+\b,-0.75-\xs);
    
    \draw [fill=white,draw=black] (1,0.75+\xs) rectangle (1+\b,0.75-\xs);
    \draw [fill=white,draw=black] (1.5,0.45+\xs) rectangle (1.5+\b,0.15-\xs);  
    \draw [fill=white,draw=black] (1.5,-0.15+\xs) rectangle (1.5+\b,-0.45-\xs);
    \draw [fill=white,draw=black] (2,-0.75+\xs) rectangle (2+\b,-0.75-\xs);
    \draw [fill=white,draw=black] (2,0.75+\xs) rectangle (2+\b,0.15-\xs);

    \draw [fill=white,draw=black] (2.5,-0.45+\xs) rectangle (2.5+\b,-0.75-\xs);  
    
    \draw [fill=white,draw=black] (3,0.75+\xs) rectangle (3+\b,0.45-\xs); 
    
    \draw [fill=white,draw=black] (3.5,0.45+\xs) rectangle (3.5+\b,0.15-\xs);
    \draw [fill=white,draw=black] (3.5,-0.45+\xs) rectangle (3.5+\b,-0.75-\xs);

    \draw [fill=white,draw=black] (4.5,-0.15+\xs) rectangle (4.5+\b,-0.45-\xs);
    \draw [fill=white,draw=black] (4.5,0.75+\xs) rectangle (4.5+\b,0.45-\xs);  
    
    \draw [fill=white,draw=black] (5,0.15+1.6*\xs) rectangle (5+\b,0.15-1.6*\xs);
    \draw [fill=white,draw=black] (5,0.45+1.6*\xs) rectangle (5+\b,0.45-1.6*\xs); 
    \draw [fill=white,draw=black] (5,0.75+1.6*\xs) rectangle (5+\b,0.75-1.6*\xs); 
    \draw [fill=white,draw=black] (5,-0.15+1.6*\xs) rectangle (5+\b,-0.15-1.6*\xs);
    \draw [fill=white,draw=black] (5,-0.45+1.6*\xs) rectangle (5+\b,-0.45-1.6*\xs); 
    \draw [fill=white,draw=black] (5,-0.75+1.6*\xs) rectangle (5+\b,-0.75-1.6*\xs);               
     \draw[thin] (5+\b-0.05,0.15-0.07) arc (0:180:0.1);
     \draw[thin,->] (5+\b-0.05-0.11,0.15-0.07) -- (5+\b-0.05,0.15-0.07+0.15); 
     \draw[thin] (5+\b-0.05,0.45-0.07) arc (0:180:0.1);
     \draw[thin,->] (5+\b-0.05-0.11,0.45-0.07) -- (5+\b-0.05,0.45-0.07+0.15);   
     \draw[thin] (5+\b-0.05,0.75-0.07) arc (0:180:0.1);
     \draw[thin,->] (5+\b-0.05-0.11,0.75-0.07) -- (5+\b-0.05,0.75-0.07+0.15);     
     \draw[thin] (5+\b,0.15+\s) -- (5+\b+7*\s,0.15+\s);
     \draw[thin] (5+\b,0.15-\s) -- (5+\b+7*\s,0.15-\s);     
     \draw[thin] (5+\b,0.45+\s) -- (5+\b+7*\s,0.45+\s);
     \draw[thin] (5+\b,0.45-\s) -- (5+\b+7*\s,0.45-\s);   
     \draw[thin] (5+\b,0.75+\s) -- (5+\b+7*\s,0.75+\s);
     \draw[thin] (5+\b,0.75-\s) -- (5+\b+7*\s,0.75-\s);   

     \draw[thin] (5+\b-0.05,-0.15-0.07) arc (0:180:0.1);
     \draw[thin,->] (5+\b-0.05-0.11,-0.15-0.07) -- (5+\b-0.05,-0.15-0.07+0.15); 
     \draw[thin] (5+\b-0.05,-0.45-0.07) arc (0:180:0.1);
     \draw[thin,->] (5+\b-0.05-0.11,-0.45-0.07) -- (5+\b-0.05,-0.45-0.07+0.15);   
     \draw[thin] (5+\b-0.05,-0.75-0.07) arc (0:180:0.1);
     \draw[thin,->] (5+\b-0.05-0.11,-0.75-0.07) -- (5+\b-0.05,-0.75-0.07+0.15);     
     \draw[thin] (5+\b,-0.15+\s) -- (5+\b+7*\s,-0.15+\s);
     \draw[thin] (5+\b,-0.15-\s) -- (5+\b+7*\s,-0.15-\s);     
     \draw[thin] (5+\b,-0.45+\s) -- (5+\b+7*\s,-0.45+\s);
     \draw[thin] (5+\b,-0.45-\s) -- (5+\b+7*\s,-0.45-\s);   
     \draw[thin] (5+\b,-0.75+\s) -- (5+\b+7*\s,-0.75+\s);
     \draw[thin] (5+\b,-0.75-\s) -- (5+\b+7*\s,-0.75-\s);       
                          
%%%%%%%%%%%%%%%%%%%%%%%%%%%%%%%%%
\node at (5/2+\b/2-0.2/2+\x/2,0) {$\implies$};
%%%%%%%%%%%%%%%%%%%%%%%%%%%%%%%%
     \draw [fill=gray!15,draw=none] (-0.05+\x,0.75+2.5*\xs) rectangle (5.4+\x,0+0.25*\xs);   
     \draw [fill=blue!15,draw=none] (-0.05+\x,-0.25*\xs) rectangle (5.4+\x,-0.75-2.5*\xs);       
    
    \draw (-0.1+\x,0.45) -- (5+\x,0.45);   
    \draw (-0.1+\x,0.75) -- (5+\x,0.75);  
    \draw (-0.1+\x,0.15) -- (5+\x,0.15);
    \draw (-0.1+\x,-0.15) -- (5+\x,-0.15);
    \draw (-0.1+\x,-0.45) -- (5+\x,-0.45);   
    \draw (-0.1+\x,-0.75) -- (5+\x,-0.75);   
    \draw [fill=GreenYellow,draw=black] (1+\x,0.15+\xs) rectangle (1+\x+\b,0.15-\xs); 
    \draw [fill=GreenYellow,draw=black] (1+\x,-0.15+\xs) rectangle (1+\x+\b,-0.15-\xs);     
    \draw [fill=GreenYellow,draw=black] (3+\x,0.15+\xs) rectangle (3+\x+\b,0.15-\xs); 
    \draw [fill=GreenYellow,draw=black] (3+\x,-0.15+\xs) rectangle (3+\x+\b,-0.45-\xs);     
    \draw [fill=GreenYellow,draw=black] (4+\x,0.75+\xs) rectangle (4+\x+\b,0.15-\xs); 
    \draw [fill=GreenYellow,draw=black] (4+\x,-0.15+\xs) rectangle (4+\x+\b,-0.15-\xs);     
    
    \draw [fill=white,draw=black] (0.5+\x,0.75+\xs) rectangle (0.5+\b+\x,0.45-\xs);
    \draw [fill=white,draw=black] (0.0+\x,0.15+\xs) rectangle (0.0+\b+\x,0.15-\xs);   
    \draw [fill=white,draw=black] (0.0+\x,-0.15+\xs) rectangle (0.0+\b+\x,-0.45-\xs);
    \draw [fill=white,draw=black] (0.5+\x,-0.45+\xs) rectangle (0.5+\b+\x,-0.75-\xs);
    
    \draw [fill=white,draw=black] (1+\x,0.75+\xs) rectangle (1+\b+\x,0.75-\xs);
    \draw [fill=white,draw=black] (1.5+\x,0.45+\xs) rectangle (1.5+\b+\x,0.15-\xs);  
    \draw [fill=white,draw=black] (1.5+\x,-0.15+\xs) rectangle (1.5+\b+\x,-0.45-\xs);
    \draw [fill=white,draw=black] (2+\x,-0.75+\xs) rectangle (2+\b+\x,-0.75-\xs);
    \draw [fill=white,draw=black] (2+\x,0.75+\xs) rectangle (2+\b+\x,0.15-\xs);

    \draw [fill=white,draw=black] (2.5+\x,-0.45+\xs) rectangle (2.5+\b+\x,-0.75-\xs);  
    
    \draw [fill=white,draw=black] (3+\x,0.75+\xs) rectangle (3+\b+\x,0.45-\xs); 
    
    \draw [fill=white,draw=black] (3.5+\x,0.45+\xs) rectangle (3.5+\b+\x,0.15-\xs);
    \draw [fill=white,draw=black] (3.5+\x,-0.45+\xs) rectangle (3.5+\b+\x,-0.75-\xs);

    \draw [fill=white,draw=black] (4.5+\x,-0.15+\xs) rectangle (4.5+\b+\x,-0.45-\xs);
    \draw [fill=white,draw=black] (4.5+\x,0.75+\xs) rectangle (4.5+\b+\x,0.45-\xs);  
    
    \draw [fill=white,draw=black] (5+\x,0.15+1.6*\xs) rectangle (\x+5+\b,0.15-1.6*\xs);
    \draw [fill=white,draw=black] (\x+5,0.45+1.6*\xs) rectangle (\x+5+\b,0.45-1.6*\xs); 
    \draw [fill=white,draw=black] (\x+5,0.75+1.6*\xs) rectangle (\x+5+\b,0.75-1.6*\xs); 
    \draw [fill=white,draw=black] (\x+5,-0.15+1.6*\xs) rectangle (\x+5+\b,-0.15-1.6*\xs);
    \draw [fill=white,draw=black] (\x+5,-0.45+1.6*\xs) rectangle (\x+5+\b,-0.45-1.6*\xs); 
    \draw [fill=white,draw=black] (\x+5,-0.75+1.6*\xs) rectangle (\x+5+\b,-0.75-1.6*\xs);               
     \draw[thin] (\x+5+\b-0.05,0.15-0.07) arc (0:180:0.1);
     \draw[thin,->] (\x+5+\b-0.05-0.11,0.15-0.07) -- (\x+5+\b-0.05,0.15-0.07+0.15); 
     \draw[thin] (\x+5+\b-0.05,0.45-0.07) arc (0:180:0.1);
     \draw[thin,->] (\x+5+\b-0.05-0.11,0.45-0.07) -- (\x+5+\b-0.05,0.45-0.07+0.15);   
     \draw[thin] (\x+5+\b-0.05,0.75-0.07) arc (0:180:0.1);
     \draw[thin,->] (\x+5+\b-0.05-0.11,0.75-0.07) -- (\x+5+\b-0.05,0.75-0.07+0.15);     
     \draw[thin] (\x+5+\b,0.15+\s) -- (\x+5+\b+7*\s,0.15+\s);
     \draw[thin] (\x+5+\b,0.15-\s) -- (\x+5+\b+7*\s,0.15-\s);     
     \draw[thin] (\x+5+\b,0.45+\s) -- (\x+5+\b+7*\s,0.45+\s);
     \draw[thin] (\x+5+\b,0.45-\s) -- (\x+5+\b+7*\s,0.45-\s);   
     \draw[thin] (\x+5+\b,0.75+\s) -- (\x+5+\b+7*\s,0.75+\s);
     \draw[thin] (\x+5+\b,0.75-\s) -- (\x+5+\b+7*\s,0.75-\s);   

     \draw[thin] (\x+5+\b-0.05,-0.15-0.07) arc (0:180:0.1);
     \draw[thin,->] (\x+5+\b-0.05-0.11,-0.15-0.07) -- (\x+5+\b-0.05,-0.15-0.07+0.15); 
     \draw[thin] (\x+5+\b-0.05,-0.45-0.07) arc (0:180:0.1);
     \draw[thin,->] (\x+5+\b-0.05-0.11,-0.45-0.07) -- (\x+5+\b-0.05,-0.45-0.07+0.15);   
     \draw[thin] (\x+5+\b-0.05,-0.75-0.07) arc (0:180:0.1);
     \draw[thin,->] (\x+5+\b-0.05-0.11,-0.75-0.07) -- (\x+5+\b-0.05,-0.75-0.07+0.15);     
     \draw[thin] (\x+5+\b,-0.15+\s) -- (\x+5+\b+7*\s,-0.15+\s);
     \draw[thin] (\x+5+\b,-0.15-\s) -- (\x+5+\b+7*\s,-0.15-\s);     
     \draw[thin] (\x+5+\b,-0.45+\s) -- (\x+5+\b+7*\s,-0.45+\s);
     \draw[thin] (\x+5+\b,-0.45-\s) -- (\x+5+\b+7*\s,-0.45-\s);   
     \draw[thin] (\x+5+\b,-0.75+\s) -- (\x+5+\b+7*\s,-0.75+\s);
     \draw[thin] (\x+5+\b,-0.75-\s) -- (\x+5+\b+7*\s,-0.75-\s);

     \draw [fill=gray!15,draw=none] (0.1,1+0.2) rectangle (0.1+0.25,1+0.25+0.2);
     \draw [fill=blue!15,draw=none] (1.6,1+0.2) rectangle (1.6+0.25,1+0.25+0.2);     
     \node[gray] at (0.55,1.325) {\footnotesize{$\bar A$}};
     \node[blue] at (2.05,1.325) {\footnotesize{$\bar B$}};
   
    \end{tikzpicture}
    \caption{The nonlocal circuit on the left can be simulated with local circuits on the right using quasiprobability simulation. If the optimal quasiprobability simulation is performed on each gate $U_1,U_2,U_3$ individually, then the total simulation overhead is given by $\gamma_S(U_1)^2 \gamma_S(U_2)^2 \gamma_S(U_3)^2$, as described in the main text.}
    \label{fig_cutting}
\end{figure}

In principle, quasiprobabilistic circuit knitting does not require any sort of classical communication between the two smaller quantum computers during the execution of the circuit.
The main question of this work is whether allowing the two computers to exchange classical information during the circuit execution can reduce the sampling overhead.
We consider three settings:
\begin{enumerate}
  \item \textbf{Local operations ($\LO$)}: The two computers can only realize operations in a product form $\cA \otimes \cB$ where $\cA$ and $\cB$ act locally on $\bar A$ and $\bar B$, respectively.
  \item \textbf{Local operations and one-way classical communication ($\LOCCOW$)}: The two computers can realize protocols that contain local operations from $\LO$ as well as classical communication from $\bar A$ to $\bar B$.
\item  \textbf{Local operations and classical communication ($\LOCC$)}: The two computers can realize protocols that contain local operations from $\LO$ as well as two-way classical communication between $\bar A$ and $\bar B$.
\end{enumerate}
Note that in the $\LO$ and \smash{$\LOCCOW$} settings, one does not necessarily require two separate quantum computers. Instead, one can run the two subcircuits in sequence on the same device.
The classical communication in the \smash{$\LOCCOW$} setting can then be simulated by classically storing the bits sent from $\bar{A}$ to $\bar{B}$. 
In contrast, the $\LOCC$ setting does require two quantum computers that exchange classical information in both directions.

To quasiprobabilistically simulate a nonlocal gate corresponding to the unitary channel $\mathcal{U}$, one requires a so-called \emph{quasiprobability decomposition} (QPD)
\begin{align} \label{eq_QPD_intro}
    \mathcal{U}=\sum_{i} a_i \mathcal{F}_i \, ,
\end{align}
where $\mathcal{F}_i$ are operations that our hardware can realize, i.e., $\mathcal{F}_i\in S$ where $S=\LO$, $S=\LOCCOW$ or $S=\LOCC$ depending on the considered setting. The coefficients $a_i$ are real numbers which might take negative values (hence the ``quasi" in quasiprobability).
During the circuit execution, the gate $\mathcal{U}$ gets randomly replaced by one of the gates $\mathcal{F}_i$.
The sampling overhead of the quasiprobability simulation is given by $\kappa^2$ where $\kappa$ is the one-norm of the coefficients, i.e.~$\kappa\coloneqq\sum_i |a_i|$.
Therefore it is desirable to find an optimal quasiprobability decomposition with smallest possible $\kappa$.
The smallest achievable $\kappa$ for a gate $U$ described in the setting \smash{$S\in\{\LO,\LOCCOW,\LOCC\}$} is denoted by $\gamma_{S}(U)$ and we refer to it as the $\gamma$-factor of $U$ (see~\cref{def_gamma_factor}).
Since the $\LOCC$ setting is strictly more powerful than \smash{$\LOCCOW$}, which in turn is itself more powerful than $\LO$, we have
\begin{align} \label{eq_trivial}
    \gamma_{\LOCC}(U) \leq \gamma_{\LOCCOW}(U) \leq \gamma_{\LO}(U) \, .
\end{align}

When simulating a single nonlocal gate $U$ via an optimal QPD, the number of samples required to achieve a fixed accuracy increases by $\gamma_S(U)^2$ compared to directly running the nonlocal gate on a large quantum computer~\cite{TBG17,endo18,Piv_masterThesis}.
If the circuit contains $n$ nonlocal gates $U_1,\dots,U_n$ and the optimal quasiprobability simulation is performed separately for each of these gates, the overall sampling overhead is given by $\prod_{i=1}^n \gamma_S(U_i)^2$.
For $n$ identical nonlocal gates $U$, this sampling overhead scales as $\gamma_S(U)^{2n}=\exp(O(n))$, which emphasizes that the number of nonlocal gates $n$ may not be too large. At the same time, it highlights the importance of choosing a setting $S$ where the corresponding $\gamma$-factor $\gamma_S(U)$ is small.
In this work, we ask the question how much classical communication can help for circuit knitting.
This requires a good understanding of the optimal sampling overhead $\gamma_S$ for the three settings \smash{$S\in\{\LO,\LOCCOW,\LOCC\}$} and how they differ.  

\paragraph{Results}
Computing the $\gamma$-factor for a given unitary is a nontrivial task as it is given via a complicated optimization problem (see~\cref{def_gamma_factor}). 
We show  that for a large class of two-qubit unitaries $U$, including all Clifford gates as well as certain non-Clifford gates such as controlled rotation gates $\mathrm{CR}_X(\theta),\mathrm{CR}_Y(\theta),\mathrm{CR}_Z(\theta)$\footnote{For $\sigma\in\{X,Y,Z\}$, $\mathrm{CR}_{\sigma}(\theta)$ denotes a controlled $\mathrm{R}_{\sigma}(\theta)$ gate with rotation angle $\theta$, where $\mathrm{R}_{\sigma}(\theta)\coloneqq \ee^{-\ci\frac{\theta}{2}\sigma}$.} and two-qubit rotations, $\mathrm{R}_{XX}(\theta),\mathrm{R}_{YY}(\theta),\mathrm{R}_{ZZ}(\theta)$\footnote{For $\sigma\in\{X,Y,Z\}$, the $\mathrm{R}_{\sigma\sigma}(\theta)$  gate is defined by the two-qubit unitary $\ee^{-\ci \frac{\theta}{2} \sigma\otimes\sigma}$ for some real number $\theta$.} there is no advantage in having classical communication when quasiprobabilistically simulating a \emph{single instance} of the gate with local operations, that is
\begin{align}\label{eq_intro_equality}
    \gamma_{\LOCC}(U) = \gamma_{\LOCCOW}(U) = \gamma_{\LO}(U) \, .
\end{align}
In fact, we prove a closed-form expression for the $\gamma$-factor as shown in~\cref{thm_main} and~\cref{cor_bounds_match}. This is the first time that an exact characterization of the optimal sampling overhead is found, as previous works~\cite{Mitarai_2021,MF_21} only showed upper bounds.
\cref{tab_results_gamma} gives an overview of gates for which we provide an analytical formula for the $\gamma$-factor under $\LO$ and $\LOCC$.

\begin{table}[!htb]
\centering
\bgroup
\def\arraystretch{1.4}
  \begin{tabular}{V{2.5}c V{2.5} c|cV{2.5}c|cV{2.5}}
  \clineB{2-5}{2.5}  
\multicolumn{1}{cV{2.5}}{}&
      \multicolumn{2}{cV{2.5}}{Clifford gates} &
      \multicolumn{2}{cV{2.5}}{non-Clifford gates} \\ \cline{2-5}
\multicolumn{1}{cV{2.5}}{}& \multirow{ 2}{*}{two-qubit gates} & $n$-qubit gates& gates defined in~\cref{thm_main}      & \multirow{ 2}{*}{others}  \\ 
\multicolumn{1}{cV{2.5}}{}&                                   & for $n\geq 3$  & e.g.~$\mathrm{CR}_Z(\theta)$ $\&$ $\mathrm{R}_{ZZ}(\theta)$ & \\
\clineB{1-5}{2.5}
    $\gamma_{\LO}$ &\cref{cor_bounds_match}  & unknown  & \cref{thm_main} $\&$~\cref{cor_bounds_match} & unknown \\
\clineB{1-5}{2.5}
  $\gamma_{\LOCC}$ &\cref{cor_bounds_match}  & \cref{thm_clifford}    &  \cref{thm_main} $\&$~\cref{cor_bounds_match}  & unknown \\
\clineB{1-5}{2.5}
  \end{tabular}
\egroup
		\caption{Overview of nonlocal gates for which we have an analytic understanding of the optimal sampling overheads $\gamma_{\LO}$ and $\gamma_{\LOCC}$.}
		\label{tab_results_gamma}
	\end{table}

At first sight,~\cref{eq_intro_equality} might seem discouraging, as it suggests that classical communication does not help for quasiprobabilistic simulation of nonlocal gates.
However, the statement only asserts that classical communication provides no advantage for simulating a \emph{single} instance of the gate $U$ and things change drastically when considering a circuit with many instances of the same nonlocal gate.

We present a novel technique utilizing bidirectional classical communication between $\bar A$ and $\bar B$ to reduce the overall sampling overhead for a circuit containing $n$ instances of some nonlocal Clifford gate $U$ from $\gamma_{\LOCC}(U)^{2n}$ to \smash{$\gamma_{\LOCC}^{(n)}(U)^{2n}$} where
\begin{align*}
    \gamma_{\LOCC}^{(n)}(U) \coloneqq  \gamma_{\LOCC}(U^{\otimes n})^{1/n} \, .
\end{align*}
We prove that $\gamma_{\LOCC}^{(n)}(U) < \gamma_{\LOCC}(U)$ for any entangling Clifford gate $U$.
Put differently, we show that utilizing the optimal quasiprobability simulation for each nonlocal gate individually is not optimal, and we propose a method that can significantly reduce the sampling overhead by making use of classical communication.

For example, the overhead of simulating $n$ nonlocal $\CNOT$ gates is reduced from $\cO(9^n)$ to $\cO(4^n)$ since $\gamma_{\LOCC}(\CNOT)=3$ but \smash{$\gamma_{\LOCC}^{(n)}(\CNOT)=(2^{n+1}-1)^{1/n}$} and thus \smash{$\gamma_{\LOCC}^{(n)}(\CNOT)^{2n}=\cO(4^n)$}.
By a similar argument, the sampling overhead of the $\SWAP$ is reduced from $\cO(49^n)$ to $\cO(16^n)$.
If one restricts oneself to the \smash{$\LOCCOW$} setting instead of $\LOCC$, we show that a reduction of sampling overhead is still possible, albeit not as strongly as above.
More specifically, we show that for $\CNOT$ gates under the \smash{$\LOCCOW$} setting the sampling overhead can be reduced to $O(8^n)$, which lies in between the $\LO$ and the $\LOCC$ scenarios.
% \begin{table}[!htb]
% \centering
% \bgroup
% \def\arraystretch{1.5}
%   \begin{tabular}{V{2.5}cV{2.5}cV{1.0}cV{2.5}}
%   \clineB{2-3}{2.5}  
% \multicolumn{1}{cV{2.5}}{}& \multicolumn{1}{cV{1.0}}{$\CNOT$} &  \multicolumn{1}{cV{2.5}}{$\SWAP$} \\ \clineB{1-3}{2.5}
% $\LO$  & $O(9^n)$ &  $O(49^n)$ \\\clineB{1-3}{1.0}
% $\LOCCOW$  & $O(8^n)$ & $O(32^n)$  \\\clineB{1-3}{1.0}
% $\LOCC$  & $O(4^n)$ & $O(16^n)$  \\\clineB{1-3}{2.5}
%   \end{tabular}
% \egroup
% 		\caption{Asymptotic scaling of the sampling overhead for a circuit with $n$ nonlocal $\CNOT$ or $\SWAP$ gates that are simulated with local operations with or without classical communication.}
% 		\label{tab_CNOT_SWAP}
% 	\end{table}

The central idea behind our technique is to realize the desired unitary via gate teleportation~\cite{GC99}.
For instance, a $\CNOT$ gate can be realized by consuming a preexisting Bell pair $\ket{\Psi}$ shared between $A$ and $B$ using a simple $\LOCC$ protocol, as depicted in~\cref{fig_CNOT}.
In case we have $n$ nonlocal $\CNOT$ gates we can generate all the required $n$ Bell pairs at the same time, at the cost of an additional memory overhead (see~\cref{fig_CNOT}).
Interestingly, a joint local simulation of $n$ Bell pairs is considerably cheaper than locally simulating $n$-times a single Bell pair. 
In technical terms this is captured by a strict submultiplicativity of the $\gamma$-factor, i.e.,
\begin{align} \label{eq_submult}
    \gamma_{\LOCC}(\ket{\Psi}^{\otimes n}) = 2^{n+1}-1 < 3^n = \gamma_{\LOCC}(\ket{\Psi})^n \, ,
\end{align}
for $n>1$.
The formal statement justifying~\cref{eq_submult} can be found in~\cref{corr_gamma_ebit}.
Since the gate teleportation protocol works for arbitrary Clifford gates, our method can be generalized accordingly and reduce the sampling overhead for arbitrary Clifford gates (see~\cref{thm_clifford}).
We also briefly discuss controlled-rotation gates as example of non-Clifford gates, and show that an adapted variant of our technique can also reduce the sampling overhead in certain parameter regimes.

\begin{figure}[htb!]
    \centering
    \begin{tikzpicture}[thick,scale=0.8]
    \def\x{3.25}
    \def\xx{3.25}
    \def\xf{10.25}
    \def\cc{0.02}
    \def \g{1}
    \def \gg{1.8}
    \def \s{7.8}

     \draw [fill=gray!15,draw=none] (0.05,0.06) rectangle (2.45,1.2);   
     \draw [fill=blue!15,draw=none] (0.05,-0.06) rectangle (2.45,-1.2);  

     \draw [fill=gray!15,draw=none] (\x+1.1,0.06) rectangle (\x+7.1,1.2);   
     \draw [fill=blue!15,draw=none] (\x+1.1,-0.06) rectangle (\x+7.1,-1.2); 

     \draw [fill=gray!15,draw=none] (\x+1.1+\s,0.06) rectangle (\x+7.1+\s,1.2);   
     \draw [fill=blue!15,draw=none] (\x+1.1+\s,-0.06) rectangle (\x+7.1+\s,-1.2);

     \draw [fill=gray!15,draw=none] (0.1,1+0.4) rectangle (0.1+0.25,1+0.25+0.4);
     \draw [fill=blue!15,draw=none] (1.6,1+0.4) rectangle (1.6+0.25,1+0.25+0.4);     
     \node[gray] at (0.6,1.125+0.4) {\footnotesize{$A$}};
     \node[blue] at (2.1,1.125+0.4) {\footnotesize{$B$}};
    
     \draw (0,0.9) -- (1,0.9);
     \draw (0,-0.9) -- (1,-0.9);     
     \draw[fill=black] (0.5,0.9) circle (0.5mm);
     \draw (0.5,-0.9) circle (1.5mm);
     \draw (0.5,0.9) -- (0.5,-0.9-0.15);
     \node at (1.25,0) {$\ldots$};
     \draw (1.5,0.9) -- (2.5,0.9);
     \draw (1.5,-0.9) -- (2.5,-0.9);     
     \draw[fill=black] (2,0.9) circle (0.5mm);
     \draw (2,-0.9) circle (1.5mm);
     \draw (2,0.9) -- (2,-0.9-0.15);
%%%     
     \node at (\x/2+\g/2+2.5/2,0) {$=$};
%%%%%%%%%%%%%%%%%%%%%%%%%%%%%%%%%%%%%%%%%%%%%%%%%%%%
\draw (\x+\g,0.9) -- (\x+\g+2.25,0.9);
\draw (\g+\x,-0.9) -- (\g+\x+2.25,-0.9);
\draw (\g+\x+2.25+0.4,0.9) -- (\g+\x+2.25+0.7,0.9);
\draw (\g+\x+2.25+0.4,-0.9) -- (\g+\x+2.25+0.7,-0.9);
\draw[fill=black] (\g+\x+0.35,0.4) circle (0.5mm);
\draw[fill=black] (\g+\x+0.35,-0.4) circle (0.5mm);
\draw[decorate,decoration={snake,amplitude=.4mm,segment length=2mm,post length=1mm}] (\g+\x+0.35,0.4) -- (\g+\x+0.35,-0.4);
\draw (\g+\x+0.85,0.4) circle (1.5mm);
\draw[fill=black] (\g+\x+0.85,0.9) circle (0.5mm);
\draw (\g+\x+0.85,0.9) -- (\g+\x+0.85,0.25);
\draw (\g+\x+0.85,-0.9) circle (1.5mm);
\draw[fill=black] (\g+\x+0.85,-0.4) circle (0.5mm);
\draw (\g+\x+0.85,-1.05) -- (\g+\x+0.85,-0.4);
\draw (\g+\x+0.35,0.4) -- (\g+\x+1.3,0.4);
\draw (\g+\x+0.35,-0.4) -- (\g+\x+1.3,-0.4);
\draw (\g+\x+1.3,0.4-0.2) rectangle (\g+\x+1.3+0.5,0.4+0.2);
\draw (\g+\x+1.3,-0.4-0.2) rectangle (\g+\x+1.3+0.5,-0.4+0.2);
\draw (\g+\x+2.25,0.9-0.2) rectangle (\g+\x+2.25+0.4,0.9+0.2);
\draw (\g+\x+2.25,-0.9-0.2) rectangle (\g+\x+2.25+0.4,-0.9+0.2);
\node at (\g+\x+2.25+0.2,0.9) {\footnotesize{$Z$}};
\node at (\g+\x+2.25+0.2,-0.9) {\footnotesize{$X$}};
\draw[thin,->] (\g+\x+1.3+0.25+0.05,0.4-0.1) -- (\g+\x+1.3+0.25+0.1+0.05,0.4+0.1);
\node[] at (\g+\x+1.3+0.12,0.4+0.08) {\tiny{$Z$}};
\draw[thin] (\g+\x+1.3+0.25+0.1+0.07,0.4-0.1) arc (0:180:0.1);
\draw[thin,->] (\g+\x+1.3+0.25+0.05,-0.4-0.1) -- (\g+\x+1.3+0.25+0.1+0.05,-0.4+0.1);
\node[] at (\g+\x+1.3+0.13,-0.4+0.08) {\tiny{$X$}};
\draw[thin] (\g+\x+1.3+0.25+0.1+0.07,-0.4-0.1) arc (0:180:0.1); 
\draw[thin] (\g+\x+1.3+0.5,0.4+\cc) -- (\g+\x+2.25+0.2+\cc-\cc-\cc-\cc-\cc+0.0088,0.4+\cc);
\draw[thin] (\g+\x+1.3+0.5,0.4-\cc) -- (\g+\x+2.25+0.2-\cc-\cc-\cc-\cc-\cc+0.0088,0.4-\cc);
\draw[thin] (\g+\x+2.25+0.2+\cc-\cc-\cc-\cc-\cc,0.4+\cc) -- (\g+\x+2.25+0.2+\cc-\cc-\cc-\cc-\cc,-0.7);
\draw[thin] (\g+\x+2.25+0.2-\cc-\cc-\cc-\cc-\cc,0.4-\cc) -- (\g+\x+2.25+0.2-\cc-\cc-\cc-\cc-\cc,-0.7);
\draw[thin] (\g+\x+1.3+0.5,-0.4+\cc) -- (\g+\x+2.25+0.2-\cc+\cc+\cc+\cc+\cc+0.0088,-0.4+\cc);
\draw[thin] (\g+\x+1.3+0.5,-0.4-\cc) -- (\g+\x+2.25+0.2+\cc+\cc+\cc+\cc+\cc+0.0088,-0.4-\cc);
\draw[thin] (\g+\x+2.25+0.2+\cc+\cc+\cc+\cc+\cc,-0.4-\cc+0.0088) -- (\g+\x+2.25+0.2+\cc+\cc+\cc+\cc+\cc,0.7);
\draw[thin] (\g+\x+2.25+0.2-\cc+\cc+\cc+\cc+\cc,-0.4+\cc+0.0088) -- (\g+\x+2.25+0.2-\cc+\cc+\cc+\cc+\cc,0.7);

\node at (\g+\x/2+2.25/2+0.7/2+\x/2+\xx/2,0) {$\ldots$};
\draw (\g+\x+\xx,0.9) -- (\g+\x+2.25+\xx,0.9);
\draw (\g+\x+\xx,-0.9) -- (\g+\x+2.25+\xx,-0.9);
\draw (\g+\x+2.25+0.4+\xx,0.9) -- (\g+\x+2.25+0.7+\xx,0.9);
\draw (\g+\x+2.25+0.4+\xx,-0.9) -- (\g+\x+2.25+0.7+\xx,-0.9);
\draw[fill=black] (\g+\x+0.35+\xx,0.4) circle (0.5mm);
\draw[fill=black] (\g+\x+0.35+\xx,-0.4) circle (0.5mm);
\draw[decorate,decoration={snake,amplitude=.4mm,segment length=2mm,post length=1mm}] (\g+\x+0.35+\xx,0.4) -- (\g+\x+0.35+\xx,-0.4);
\draw (\g+\x+0.85+\xx,0.4) circle (1.5mm);
\draw[fill=black] (\g+\x+0.85+\xx,0.9) circle (0.5mm);
\draw (\g+\x+0.85+\xx,0.9) -- (\g+\x+0.85+\xx,0.25);
\draw (\g+\x+0.85+\xx,-0.9) circle (1.5mm);
\draw[fill=black] (\g+\x+0.85+\xx,-0.4) circle (0.5mm);
\draw (\g+\x+0.85+\xx,-1.05) -- (\g+\x+0.85+\xx,-0.4);
\draw (\g+\x+0.35+\xx,0.4) -- (\g+\x+1.3+\xx,0.4);
\draw (\g+\x+0.35+\xx,-0.4) -- (\g+\x+1.3+\xx,-0.4);
\draw (\g+\x+1.3+\xx,0.4-0.2) rectangle (\g+\x+1.3+0.5+\xx,0.4+0.2);
\draw (\g+\x+1.3+\xx,-0.4-0.2) rectangle (\g+\x+1.3+0.5+\xx,-0.4+0.2);
\draw (\g+\x+2.25+\xx,0.9-0.2) rectangle (\g+\x+2.25+0.4+\xx,0.9+0.2);
\draw (\g+\x+2.25+\xx,-0.9-0.2) rectangle (\g+\x+2.25+0.4+\xx,-0.9+0.2);
\node at (\g+\x+2.25+0.2+\xx,0.9) {\footnotesize{$Z$}};
\node at (\g+\x+2.25+0.2+\xx,-0.9) {\footnotesize{$X$}};
\draw[thin,->] (\g+\x+1.3+0.25+0.05+\xx,0.4-0.1) -- (\g+\x+1.3+0.25+0.1+0.05+\xx,0.4+0.1);
\node[] at (\g+\x+1.3+0.12+\xx,0.4+0.08) {\tiny{$Z$}};
\draw[thin] (\g+\x+1.3+0.25+0.1+0.07+\xx,0.4-0.1) arc (0:180:0.1);
\draw[thin,->] (\g+\x+1.3+0.25+0.05+\xx,-0.4-0.1) -- (\g+\x+1.3+0.25+0.1+0.05+\xx,-0.4+0.1);
\node[] at (\g+\x+1.3+0.13+\xx,-0.4+0.08) {\tiny{$X$}};
\draw[thin] (\g+\x+1.3+0.25+0.1+0.07+\xx,-0.4-0.1) arc (0:180:0.1); 
\draw[thin] (\g+\x+1.3+0.5+\xx,0.4+\cc) -- (\g+\xx+\x+2.25+0.2+\cc-\cc-\cc-\cc-\cc+0.0088,0.4+\cc);
\draw[thin] (\xx+\g+\x+1.3+0.5,0.4-\cc) -- (\xx+\g+\x+2.25+0.2-\cc-\cc-\cc-\cc-\cc+0.0088,0.4-\cc);
\draw[thin] (\xx+\g+\x+2.25+0.2+\cc-\cc-\cc-\cc-\cc,0.4+\cc) -- (\xx+\g+\x+2.25+0.2+\cc-\cc-\cc-\cc-\cc,-0.7);
\draw[thin] (\xx+\g+\x+2.25+0.2-\cc-\cc-\cc-\cc-\cc,0.4-\cc) -- (\xx+\g+\x+2.25+0.2-\cc-\cc-\cc-\cc-\cc,-0.7);
\draw[thin] (\xx+\g+\x+1.3+0.5,-0.4+\cc) -- (\xx+\g+\x+2.25+0.2-\cc+\cc+\cc+\cc+\cc+0.0088,-0.4+\cc);
\draw[thin] (\xx+\g+\x+1.3+0.5,-0.4-\cc) -- (\xx+\g+\x+2.25+0.2+\cc+\cc+\cc+\cc+\cc+0.0088,-0.4-\cc);
\draw[thin] (\xx+\g+\x+2.25+0.2+\cc+\cc+\cc+\cc+\cc,-0.4-\cc+0.0088) -- (\xx+\g+\x+2.25+0.2+\cc+\cc+\cc+\cc+\cc,0.7);
\draw[thin] (\xx+\g+\x+2.25+0.2-\cc+\cc+\cc+\cc+\cc,-0.4+\cc+0.0088) -- (\xx+\g+\x+2.25+0.2-\cc+\cc+\cc+\cc+\cc,0.7);

%%%%%%%%%%%%%%%%%%%%%%%%%%%%%%%%%%%%%%%%%%%%%%%%%%%%
\node at (\gg/2+\g/2+\x/2+2.25/2+0.7/2+\xx/2+\xf/2,0) {$=$};
%%%%%%%%%%%%%%%%%%%%%%%%%%%%%%%%%%%%%%%%%%%%%%%%%%%%
\draw (\gg+\xf,0.9) -- (\gg+\xf+2.25,0.9);
\draw (\gg+\xf,-0.9) -- (\gg+\xf+2.25,-0.9);
\draw (\gg+\xf+2.25+0.4,0.9) -- (\gg+\xf+2.25+0.7,0.9);
\draw (\gg+\xf+2.25+0.4,-0.9) -- (\gg+\xf+2.25+0.7,-0.9);
\draw[fill=black] (\gg+\xf+0.35,0.4) circle (0.5mm);
\draw[fill=black] (\gg+\xf+0.35,-0.4) circle (0.5mm);
\draw[decorate,decoration={snake,amplitude=.4mm,segment length=2mm,post length=1mm}] (\gg+\xf+0.35,0.4) -- (\gg+\xf+0.35,-0.4);
\draw (\gg+\xf+0.85,0.4) circle (1.5mm);
\draw[fill=black] (\gg+\xf+0.85,0.9) circle (0.5mm);
\draw (\gg+\xf+0.85,0.9) -- (\gg+\xf+0.85,0.25);
\draw (\gg+\xf+0.85,-0.9) circle (1.5mm);
\draw[fill=black] (\gg+\xf+0.85,-0.4) circle (0.5mm);
\draw (\gg+\xf+0.85,-1.05) -- (\gg+\xf+0.85,-0.4);
\draw (\gg+\xf+0.35,0.4) -- (\gg+\xf+1.3,0.4);
\draw (\gg+\xf+0.35,-0.4) -- (\gg+\xf+1.3,-0.4);
\draw (\gg+\xf+1.3,0.4-0.2) rectangle (\gg+\xf+1.3+0.5,0.4+0.2);
\draw (\gg+\xf+1.3,-0.4-0.2) rectangle (\gg+\xf+1.3+0.5,-0.4+0.2);
\draw (\gg+\xf+2.25,0.9-0.2) rectangle (\gg+\xf+2.25+0.4,0.9+0.2);
\draw (\gg+\xf+2.25,-0.9-0.2) rectangle (\gg+\xf+2.25+0.4,-0.9+0.2);
\node at (\gg+\xf+2.25+0.2,0.9) {\footnotesize{$Z$}};
\node at (\gg+\xf+2.25+0.2,-0.9) {\footnotesize{$X$}};
\draw[thin,->] (\gg+\xf+1.3+0.25+0.05,0.4-0.1) -- (\gg+\xf+1.3+0.25+0.1+0.05,0.4+0.1);
\node[] at (\gg+\xf+1.3+0.12,0.4+0.08) {\tiny{$Z$}};
\draw[thin] (\gg+\xf+1.3+0.25+0.1+0.07,0.4-0.1) arc (0:180:0.1);
\draw[thin,->] (\gg+\xf+1.3+0.25+0.05,-0.4-0.1) -- (\gg+\xf+1.3+0.25+0.1+0.05,-0.4+0.1);
\node[] at (\gg+\xf+1.3+0.13,-0.4+0.08) {\tiny{$X$}};
\draw[thin] (\gg+\xf+1.3+0.25+0.1+0.07,-0.4-0.1) arc (0:180:0.1); 
\draw[thin] (\g+\x+1.3+0.25+0.1+0.07,-0.4-0.1) arc (0:180:0.1); 
\draw[thin] (\gg+\xf+1.3+0.5,0.4+\cc) -- (\gg+\xf+2.25+0.2+\cc-\cc-\cc-\cc-\cc+0.0088,0.4+\cc);
\draw[thin] (\gg+\xf+1.3+0.5,0.4-\cc) -- (\gg+\xf+2.25+0.2-\cc-\cc-\cc-\cc-\cc+0.0088,0.4-\cc);
\draw[thin] (\gg+\xf+2.25+0.2+\cc-\cc-\cc-\cc-\cc,0.4+\cc) -- (\gg+\xf+2.25+0.2+\cc-\cc-\cc-\cc-\cc,-0.7);
\draw[thin] (\gg+\xf+2.25+0.2-\cc-\cc-\cc-\cc-\cc,0.4-\cc) -- (\gg+\xf+2.25+0.2-\cc-\cc-\cc-\cc-\cc,-0.7);
\draw[thin] (\gg+\xf+1.3+0.5,-0.4+\cc) -- (\gg+\xf+2.25+0.2-\cc+\cc+\cc+\cc+\cc+0.0088,-0.4+\cc);
\draw[thin] (\gg+\xf+1.3+0.5,-0.4-\cc) -- (\gg+\xf+2.25+0.2+\cc+\cc+\cc+\cc+\cc+0.0088,-0.4-\cc);
\draw[thin] (\gg+\xf+2.25+0.2+\cc+\cc+\cc+\cc+\cc,-0.4-\cc+0.0088) -- (\gg+\xf+2.25+0.2+\cc+\cc+\cc+\cc+\cc,0.7);
\draw[thin] (\gg+\xf+2.25+0.2-\cc+\cc+\cc+\cc+\cc,-0.4+\cc+0.0088) -- (\gg+\xf+2.25+0.2-\cc+\cc+\cc+\cc+\cc,0.7);
\node at (\gg/2+\xf/2+2.25/2+0.7/2+\xf/2+\xx/2+1,0) {$\ldots$};
\draw (\gg+\xf+\xx,0.9) -- (\gg+\xf+2.25+\xx,0.9);
\draw (\gg+\xf+\xx,-0.9) -- (\gg+\xf+2.25+\xx,-0.9);
\draw (\gg+\xf+2.25+0.4+\xx,0.9) -- (\gg+\xf+2.25+0.7+\xx,0.9);
\draw (\gg+\xf+2.25+0.4+\xx,-0.9) -- (\gg+\xf+2.25+0.7+\xx,-0.9);
\draw[fill=black] (\gg+\xf+0.35+0.2,0.4-0.27) circle (0.5mm);
\draw[fill=black] (\gg+\xf+0.35+0.2,-0.4+0.27) circle (0.5mm);
\draw[thin,decorate,decoration={snake,amplitude=.2mm,segment length=0.6mm,post length=0.8mm}] (\gg+\xf+0.35+0.2,0.4-0.27) -- (\gg+\xf+0.35+0.2,-0.4+0.27);
\draw (\gg+\xf+0.35+0.2,0.4-0.27) -- (\gg+\xf+0.35+0.2+2.5,0.4-0.27);
\draw (\gg+\xf+0.35+0.2,-0.4+0.27) -- (\gg+\xf+0.35+0.2+2.5,-0.4+0.27);
\draw (\gg+\xf+0.35+0.2+2.5,0.4-0.27) -- (\gg+\xf+0.35+\xx,0.4);  
\draw (\gg+\xf+0.35+0.2+2.5,-0.4+0.27) -- (\gg+\xf+0.35+\xx,-0.4);
\draw (\gg+\xf+0.85+\xx,0.4) circle (1.5mm);
\draw[fill=black] (\gg+\xf+0.85+\xx,0.9) circle (0.5mm);
\draw (\gg+\xf+0.85+\xx,0.9) -- (\gg+\xf+0.85+\xx,0.25);
\draw (\gg+\xf+0.85+\xx,-0.9) circle (1.5mm);
\draw[fill=black] (\gg+\xf+0.85+\xx,-0.4) circle (0.5mm);
\draw (\gg+\xf+0.85+\xx,-1.05) -- (\gg+\xf+0.85+\xx,-0.4);
\draw (\gg+\xf+0.35+\xx,0.4) -- (\gg+\xf+1.3+\xx,0.4);
\draw (\gg+\xf+0.35+\xx,-0.4) -- (\gg+\xf+1.3+\xx,-0.4);
\draw (\gg+\xf+1.3+\xx,0.4-0.2) rectangle (\gg+\xf+1.3+0.5+\xx,0.4+0.2);
\draw (\gg+\xf+1.3+\xx,-0.4-0.2) rectangle (\gg+\xf+1.3+0.5+\xx,-0.4+0.2);
\draw (\gg+\xf+2.25+\xx,0.9-0.2) rectangle (\gg+\xf+2.25+0.4+\xx,0.9+0.2);
\draw (\gg+\xf+2.25+\xx,-0.9-0.2) rectangle (\gg+\xf+2.25+0.4+\xx,-0.9+0.2);
\node at (\gg+\xf+2.25+0.2+\xx,0.9) {\footnotesize{$Z$}};
\node at (\gg+\xf+2.25+0.2+\xx,-0.9) {\footnotesize{$X$}};
\draw[thin,->] (\gg+\xf+1.3+0.25+0.05+\xx,0.4-0.1) -- (\gg+\xf+1.3+0.25+0.1+0.05+\xx,0.4+0.1);
\node[] at (\gg+\xf+1.3+0.12+\xx,0.4+0.08) {\tiny{$Z$}};
\draw[thin] (\gg+\xf+1.3+0.25+0.1+0.07+\xx,0.4-0.1) arc (0:180:0.1);
\draw[thin,->] (\gg+\xf+1.3+0.25+0.05+\xx,-0.4-0.1) -- (\gg+\xf+1.3+0.25+0.1+0.05+\xx,-0.4+0.1);
\node[] at (\gg+\xf+1.3+0.13+\xx,-0.4+0.08) {\tiny{$X$}};
\draw[thin] (\gg+\xf+1.3+0.25+0.1+0.07+\xx,-0.4-0.1) arc (0:180:0.1); 
\draw[thin] (\xx+\gg+\xf+1.3+0.5,0.4+\cc) -- (\xx+\gg+\xf+2.25+0.2+\cc-\cc-\cc-\cc-\cc+0.0088,0.4+\cc);
\draw[thin] (\xx+\gg+\xf+1.3+0.5,0.4-\cc) -- (\xx+\gg+\xf+2.25+0.2-\cc-\cc-\cc-\cc-\cc+0.0088,0.4-\cc);
\draw[thin] (\xx+\gg+\xf+2.25+0.2+\cc-\cc-\cc-\cc-\cc,0.4+\cc) -- (\xx+\gg+\xf+2.25+0.2+\cc-\cc-\cc-\cc-\cc,-0.7);
\draw[thin] (\xx+\gg+\xf+2.25+0.2-\cc-\cc-\cc-\cc-\cc,0.4-\cc) -- (\xx+\gg+\xf+2.25+0.2-\cc-\cc-\cc-\cc-\cc,-0.7);
\draw[thin] (\xx+\gg+\xf+1.3+0.5,-0.4+\cc) -- (\xx+\gg+\xf+2.25+0.2-\cc+\cc+\cc+\cc+\cc+0.0088,-0.4+\cc);
\draw[thin] (\xx+\gg+\xf+1.3+0.5,-0.4-\cc) -- (\xx+\gg+\xf+2.25+0.2+\cc+\cc+\cc+\cc+\cc+0.0088,-0.4-\cc);
\draw[thin] (\xx+\gg+\xf+2.25+0.2+\cc+\cc+\cc+\cc+\cc,-0.4-\cc+0.0088) -- (\xx+\gg+\xf+2.25+0.2+\cc+\cc+\cc+\cc+\cc,0.7);
\draw[thin] (\xx+\gg+\xf+2.25+0.2-\cc+\cc+\cc+\cc+\cc,-0.4+\cc+0.0088) -- (\xx+\gg+\xf+2.25+0.2-\cc+\cc+\cc+\cc+\cc,0.7);

%%%%%%%%%%%%%%%%%%%%%%%%%%%%%%%%%%%%%%%%%%%%%%%%%%%%
    \end{tikzpicture}
    \caption{Graphical explanation of how to realize two $\CNOT$ gates between two parties $A$ and $B$ in a $\LOCC$ setting via gate teleportation. The wavy line depicts a Bell pair shared between $A$ and $B$. By generating the two Bell pairs simultaneously (instead generating a single Bell pair twice), we can utilize the submultiplicativity of the $\gamma$-factor under the tensor product to reduce the total sampling overhead. 
    %In summary, we can generate the two ebits cheaper compared to generating a single ebit twice.
    }
    \label{fig_CNOT}
\end{figure}
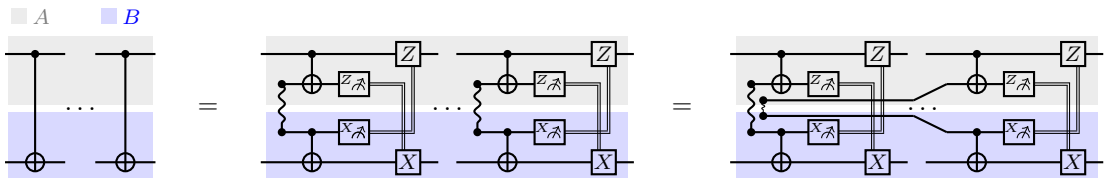

Our results, summarized in~\Cref{tab_results}, show that classical communication can be helpful to considerably reduce the sampling overhead when using circuit knitting techniques to perform nonlocal computations locally. Current research efforts in building quantum computers strongly suggest that these results may be useful in the close future as planned devices will be scaled up by connecting smaller chips, first via a classical and later via a quantum link~\cite{blogpost}.
\begin{table}[!htb]
\centering
\bgroup
\def\arraystretch{1.25}
  \begin{tabular}{V{2.5}cV{2.5}c|cV{2.5}c|cV{2.5}}
  \clineB{2-5}{2.5}  
\multicolumn{1}{cV{2.5}}{}&
      \multicolumn{2}{cV{2.5}}{Clifford gates} &
      \multicolumn{2}{cV{2.5}}{non-Clifford gates} \\ \cline{2-5}
\multicolumn{1}{cV{2.5}}{}& \multirow{ 2}{*}{two-qubit gates} & $n$-qubit gates& gates defined in~\cref{thm_main}      & \multirow{ 2}{*}{others}  \\ 
\multicolumn{1}{cV{2.5}}{}&                                   & for $n\geq 3$  &e.g.~$\mathrm{CR}_Z(\theta)$ $\&$ $\mathrm{R}_{ZZ}(\theta)$ & \\
\clineB{1-5}{2.5}
$\mathrm{CC}$ helps for  & \ding{55} & unknown  & \ding{55} & unknown  \\
\textbf{single} instances& \cref{cor_bounds_match}       &                     & \cref{thm_main} $\&$~\cref{cor_bounds_match}       &    \\
\clineB{1-5}{2.5}
$\mathrm{CC}$ helps for & $\checkmark$ & $\checkmark$ & $(\checkmark)$ &  unknown \\
\textbf{multiple} instances & \cref{thm_clifford}       &  \cref{thm_clifford}                   & \cref{sec_nonClifford_multiple} for $\mathrm{CR}_Z(\theta)$ gate  &    \\
\clineB{1-5}{2.5}
  \end{tabular}
\egroup
		\caption{Overview of nonlocal gates for which classical communication ($\mathrm{CC}$) does or does not help for circuit knitting (denoted by a $\checkmark$ or \ding{55} respectively).}
		\label{tab_results}
	\end{table}

\paragraph{Open questions}
Our work provides a powerful framework to realize circuit knitting, but it remains an open problem to identify useful applications where this framework fits well.
Ideally, the considered circuits should have an almost local structure in the sense that there are only few nonlocal gates that have to be simulated.
If one wants to use circuit knitting to demonstrate a quantum advantage, it is furthermore necessary that the local computations are themselves out of reach for a classical computer, because otherwise the knitting could be simulated classically.

Furthermore, the generalization of the presented technique to non-Clifford gates is more complicated and not well understood.
In~\cref{sec_nonClifford_multiple}, we show that our method can be adapted for controlled-rotation gates $\mathrm{CR}_X(\theta)$ and reduce the sampling overhead when $\pi/3<\theta<5\pi/3$.
It remains unclear, whether the sampling overhead could be reduced even further or at all for values of $\theta \in [0,2\pi]$ outside of the above range (see~\cref{fig_openQuestion}).

It can also be noted that~\cref{tab_results_gamma} and~\cref{tab_results} have a few empty entries corresponding to gates for which it is currently unknown what the optimal sampling overhead is as well as if classical communication helps for single or multiple instances.
Solving these problems would be interesting, however at the same time these gates can be avoided in practice by decomposing them into gates for which we have a complete understanding, at the cost of an additional overhead due to the decomposition.

\paragraph{Relation to existing circuit knitting techniques}
Circuit knitting through quasiprobability simulation of nonlocal gates has already been proposed previously~\cite{Mitarai_2021,MF_21}.
However, previous works only focused on the $\LO$ setting and did not consider classical communication between the involved parties.
Furthermore, previous work proposed explicit quasiprobability decompositions without any evidence for optimality.
We show that the quasiprobability decomposition proposed in~\cite{MF_21} are optimal for a certain class of two-qubit unitaries.
Other existing circuit knitting approaches, such as~\cite{BSS16,PHOW20}, do not precisely specify the exponential scaling and state an overhead scaling as $2^{O(n)}$, whereas our work makes a more precise statement. We quantify the exponential sampling overhead as $O(\kappa^n)$ where we aim for finding the smallest possible constant $\kappa$. For practical applications, a smaller basis of the exponential scaling can make a large difference. 
Instead of cutting individual gates, as we do in this work, it has also been proposed to perform time-like cuts~\cite{PHOW20,Mitarai_2021}.
Quasiprobability simulation also has certain advantages compared to the VQE-type entanglement forging approach~\cite{forging22}, most notably that it works for arbitrary circuits.

\paragraph{Relation to resource theory of magic}
Many ideas in our work find a counterpart in the resource theory of magic~\cite{Veitch_2012,Veitch_2014,HC17}.
For instance, we replace nonlocal operations by local operations that consume some preexisting entangled state.
This is analogous to realizing a non-Clifford gate using a Clifford gadget that consumes some preexisting magic state.
While magic is the resource required to enable nonclassical computation, entanglement is the resource required to enable nonlocal computation. 
When using quasiprobability simulation to simulate non-Clifford gates with Clifford gates, the relevant resource monotone to measure the sampling overhead is the robustness of magic.
Similarly, the sampling overhead in our scheme is captured by the robustness of entanglement instead.

\paragraph{Structure} 
\Cref{sec_circuit_knitting} formally introduces the basic concepts of circuit knitting via quasiprobability simulation.
In~\cref{sec_QPD_states} we discuss optimal QPDs for the local preparation of entangled quantum states.
We show that there is a one-to-one connection between the $\gamma$-factor and an entanglement measure called \emph{robustness of entanglement}~\cite{vidal99}.
As mentioned above, the optimal decomposition of gates can be reduced to states in some instances via gate teleportation, and therefore this section serves as a central tool for subsequent results.
In~\cref{sec_single_instance} we describe upper and lower bounds for the $\gamma$-factor of gates, which coincide for a wide range of two-qubit uniatries and hence give us a complete understanding about the optimal sampling overhead for these gates.
Finally, in~\cref{sec_how_use_submult} we explain how a circuit with many nonlocal gates can be simulated locally more efficiently using classical communication.

%%%%%%%%%%%%%%%%%%%%%%%%%%%%%%%%%%%%%%%%%%%%%%%%%%%%%%%%%%%%%%%%%%%%%%%%%%%%
%%%%%%%%%%%%%%%%%%%%%%%%%%%%%%%%%%%%%%%%%%%%%%%%%%%%%%%%%%%%%%%%%%%%%%%%%%%%
\section{Circuit knitting using quasiprobability simulation}\label{sec_circuit_knitting}
Consider two quantum systems $A$ and $B$.
Let $\mathrm{L}(A)$ denote the set of linear maps from $A$ to $A$ and  $\mathscr{S}(A)\coloneqq\mathrm{L}(\mathrm{L}(A))$ the set of linear superoperators from $A$ to $A$.
The set of completely positive trace-nonincreasing maps acting from $A$ to itself is denoted by $\TNCP(A)\subset \mathscr{S}(A)$.
Note that any general trace-nonincreasing map can be physically realized by postselection on a corresponding generalized measurement.\footnote{See~\cite{endo18} for an in-depth explanation of how to include postselection in a quasiprobability simulation.}

In the setting of quasiprobability simulation, $\TNCP(A)$ does not represent the most general operations that can be realized on a quantum computer.
In fact, some non-positive superoperators on $A$ can be simulated with some appropriate post-processing, as long as they are in the set
\begin{align}
\mathrm{D}(A) := \{\mathcal{E}\in\mathscr{S}(A) | \exists \mathcal{E}^{+},\mathcal{E}^- \in\TNCP(A) \text{ such that } \mathcal{E}=\mathcal{E}^+-\mathcal{E}^- \text{ and } \mathcal{E}^++\mathcal{E}^-\in\TNCP(A)\} \, .
\end{align}
The subtlety of considering $\mathrm{D}(A)$ instead of $\TPCP(A)$ is explained in more detail in~\cref{app_nonpositive_ops} and will not play a significant role for the results in this manuscript.
The set of local operations $\LO(A,B)$ is defined as all maps on $A\otimes B$ in a product form $\mathcal{A}\otimes\mathcal{B}$ where $\mathcal{B}\in\mathrm{D}(B)$ and $\mathcal{A}\in\mathrm{D}(A)$.

$\LOCC(A,B)$ is defined as the set of all maps that can be realized as a protocol containing only local operations in $\LO(A,B)$ and classical communication between $A$ and $B$.
Note that the formal definition of $\LOCC$ is notoriously complicated due to the possibly unbounded number of classical communication rounds and the fact that later local operations can in general depend on all the previous communication. The interested reader may consult~\cite{CLMOW14} for more details.
Finally, \smash{$\LOCCOW(A,B)$} is defined as the set of all maps that can be realized as a protocol containing only local operations in $\LO(A,B)$ and classical communication from $A$ to $B$.

Consider a quantum circuit on which we want to apply quasiprobabilistic circuit knitting.
We first group the involved qubits into two systems, denoted $\bar A$ and $\bar B$.
For simplicity, we first consider the case where there is only one single nonlocal gate $U$ acting on both $\bar A$ and $\bar B$, i.e., we can write the channel capturing the evolution of the quantum circuit as $\mathcal{G}_2\circ\mathcal{U}\circ\mathcal{G}_1$ where $\mathcal{U}$ is the unitary channel corresponding to $U$ and $\mathcal{G}_i$ only act locally on $\bar A$ and $\bar B$.
Assume that the resulting state of our circuit is measured according to some local observable $O_{\bar A}\otimes O_{\bar B}$ and that we want to determine the corresponding expectation value of the outcome $\tr{(O_{\bar A}\otimes O_{\bar B})\mathcal{G}_2\circ\mathcal{U}\circ\mathcal{G}_1(\proj{0}_{\bar A \bar B})}$ where $\proj{0}_{\bar A\bar B}$ is the initial state where all qubits are in the $\ket{0}$ state.
If we are given a QPD for $\mathcal{U}$ as shown in~\cref{eq_QPD_intro}, we can write
\begin{align}\label{eq_exp_value}
    \tr{(O_{\bar A}\otimes O_{\bar B})\mathcal{G}_2\circ\mathcal{U}\circ\mathcal{G}_1(\proj{0}_{\bar A \bar B})} = \sum_i p_i \tr{(O_{\bar A}\otimes O_{\bar B}) \mathcal{G}_2\circ\mathcal{F}_i\circ\mathcal{G}_1(\proj{0}_{\bar A \bar B})} \kappa \, \mathrm{sign}(a_i) \, ,
\end{align}
where $\kappa\coloneqq\sum_i|a_i|$ and $p_i\coloneqq|a_i|/\kappa$ is a probability distribution.
% $\proj{0}_{\bar A \bar B}$ denotes the state where all qubits are in the $\ket{0}$ state.
Following~\cref{eq_exp_value} we can estimate the expectation value of the large circuit using Monte Carlo sampling: For each shot of the circuit, we randomly replace the nonlocal gate $\mathcal{U}$ with one of the local gates $\mathcal{F}_i$ with probability $p_i$. The measurement outcome of this circuit is then weighted by $\kappa\,\mathrm{sign}(a_i)$.
By repeating this procedure many times and averaging the result, we can get an arbitrary good estimate of the desired quantity.
One can verify that the number of shots to estimate the expectation value to some fixed accuracy increases by $\kappa^2$~\cite{TBG17,endo18,Piv_masterThesis}.

The above procedure can be straightforwardly applied to circuits containing $n$ nonlocal gates $U_1,\dots,U_n$ for $n \in \N$.
For each of these $n$ gates we require a QPD and we denote the associated sampling overheads by $\kappa_1,\dots,\kappa_n$.
During each shot of the circuit that is executed, each gate gets independently randomly replaced by one of the gates in its decomposition.
The total number of samples of the circuit therefore increases by $\prod_{i=1}^n\kappa_i^2$.\footnote{Note that this total sampling overhead scales exponentially in $n$ which emphasizes that the method is designed for knitting a reasonable number of gates.}

To minimize the sampling overhead, it is desirable to use QPDs with the smallest possible sampling overhead (denoted $\kappa$ above).
Such QPDs are therefore called \emph{optimal} QPDs.
We call the smallest achievable value $\gamma$-factor.
\begin{definition} \label{def_gamma_factor}
  The $\gamma$-factor of $\mathcal{E}\in \mathscr{S}(A\otimes B)$ over \smash{$S \in\{\LO(A,B), \LOCCOW(A,B), \LOCC(A,B)\}$} is defined as
  \begin{align}\label{eq_def_gamma}
      \gamma_S(\mathcal{E}) \coloneqq \min \Big\{ \sum_{i=1}^m \lvert a_i \rvert :  \mathcal{E} = \sum\limits_{i=1}^m a_i \mathcal{F}_i \text{ where } m\geq 1, \mathcal{F}_i\in S \textnormal{ and } a_i \in \R \Big\} \, .
  \end{align}
\end{definition}
We explain in~\cref{app_min_achieved} why the minimum in~\cref{eq_def_gamma} is indeed achieved.
Since in many cases the involved systems $A$ and $B$ are clear by context, we often shorten the notation and simply write $\gamma_{\LO}(\mathcal{E})$, \smash{$\gamma_{\LOCCOW}(\mathcal{E})$}, and $\gamma_{\LOCC}(\mathcal{E})$.
As a minor abuse of notation, for a unitary acting on the system $A \otimes B$ we denote by $\gamma_S(U)$ the $\gamma$-factor of the unitary channel induced by $U$.
Note that~\cref{eq_trivial} follows directly from the definition.
The $\gamma$-factor is submultiplicative under the tensor product.
\begin{lemma}
  Let $\mathcal{E}_1\in \mathscr{S}(A_1\otimes B_1)$, $\mathcal{E}_2\in \mathscr{S}(A_2\otimes B_2)$, and $S \in$ $\{\LO(A_1\otimes A_2,B_1\otimes B_2)$, \smash{$\LOCCOW(A_1\otimes A_2$}, {$B_1\otimes B_2)$}, $\LOCC(A_1\otimes A_2,B_1\otimes B_2)\}$. Then
  \begin{align*}
      \gamma_S(\mathcal{E}_1\otimes \mathcal{E}_2) \leq \gamma_S(\mathcal{E}_1) \, \gamma_S(\mathcal{E}_2) \, .
  \end{align*}
\end{lemma}
\begin{proof}
  Let $\mathcal{E}_j=\sum_i a_{i,j}\mathcal{F}_{i,j}$ be the QPD that achieves the minimum in~\cref{def_gamma_factor} for $j=1,2$.
  This directly gives us a QPD
  \begin{align*}
      \mathcal{E}_1\otimes \mathcal{E}_2 = \sum\limits_{j_1,j_2} a_{1,j_1}a_{2,j_2}\mathcal{F}_{1,j_1}\otimes \mathcal{F}_{2,j_2} \, .
  \end{align*}
  The one-norm of these quasiprobability coefficients is precisely $\gamma_S(\mathcal{E}_1) \gamma_S(\mathcal{E}_2)$.
\end{proof}
As we will see later, for many unitaries the $\gamma$-factor is not only submultiplicative, but in fact even strictly submultiplicative under the tensor product. This property will be central in reducing the sampling overhead in our technique.
Another important property of the $\gamma$-factor is that it is invariant under local unitaries.
\begin{lemma}\label{lem_invariance_local_unitary}
  Consider unitaries $U_A\otimes U_B$ and $V_A\otimes V_B$ acting locally on $A$ and $B$, respectively.
  Denote by $\mathcal{U},\mathcal{V}\in \LO(A,B)$ the induced quantum channels acting on the system $A\otimes B$.
  For all $\mathcal{E}\in \mathscr{S}(A\otimes B)$ and \smash{$S \in\{\LO(A,B), \LOCCOW(A,B), \LOCC(A,B)\}$} we have
  \begin{align*}
      \gamma_S(\mathcal{U}\circ\mathcal{E}\circ \mathcal{V}) = \gamma_S(\mathcal{E}) \, .
  \end{align*}
\end{lemma}
\begin{proof}
Any quasiprobability decomposition of $\mathcal{E}$ of the form $\mathcal{E} = \sum_i a_i \mathcal{F}_i$ automatically induces a quasiprobability decomposition of $\mathcal{U}\circ\mathcal{E}\circ \mathcal{V}$ of the form $\mathcal{U}\circ\mathcal{E}\circ \mathcal{V} = \sum_i a_i \mathcal{U}\circ\mathcal{F}_i\circ \mathcal{V}$ with identical sampling overhead. Similarly, any quasiprobability decomposition of $\mathcal{U}\circ\mathcal{E}\circ \mathcal{V}$, i.e., $\mathcal{U}\circ\mathcal{E}\circ \mathcal{V} = \sum_i a_i \mathcal{G}_i$ generates a quasiprobability decomposition of $\mathcal{E}$ of the form $\mathcal{E}= \sum_i a_i \, \mathcal{U}^{-1}\circ\mathcal{G}_i\circ\mathcal{V}^{-1}$ with identical sampling overhead. Notice that $\mathcal{U}\circ\mathcal{F}_i\circ \mathcal{V}\in S$ and $\mathcal{U}^{-1}\circ\mathcal{G}_i\circ\mathcal{V}^{-1}\in S$.
Therefore, the quasiprobability decompositions achieving the minimum in~\cref{eq_def_gamma} for $\mathcal{E}$ and $\mathcal{U}\circ\mathcal{E}\circ \mathcal{V}$ must have identical sampling overhead.
\end{proof}

%%%%%%%%%%%%%%%%%%%%%%%%%%%%%%%%%%%%%%%%%%%%%%%%%%%%%%%%%%%%%%%%%%%%%%%%%%%%
%%%%%%%%%%%%%%%%%%%%%%%%%%%%%%%%%%%%%%%%%%%%%%%%%%%%%%%%%%%%%%%%%%%%%%%%%%%%
\section{Local quasiprobability decompositions for states} \label{sec_QPD_states}
Consider a bipartite quantum state $\rho_{AB}$. The $\gamma$-factor of $\rho_{AB}$ over $S \in\{\LO(A,B)$, \smash{$\LOCCOW(A,B)$}, $\LOCC(A,B)\}$ is defined as
\begin{align}\label{eq_def_state_qpd}
      \gamma_S(\rho_{AB}) \coloneqq \min \Big\{ \gamma_S(\mathcal{E}) : \mathcal{E}\in \mathscr{S}(A\otimes B) \text{ s.t. } \mathcal{E}(\proj{0}_{AB}) = \rho_{AB} \Big\} \, ,
\end{align}
where $\proj{0}_{AB}$ denotes some fixed product state.\footnote{In our setting where $A,B$ are multi-qubit systems, one would typically consider it to be the product state where all qubits are in the $\ket{0}$ state.}
This quantity characterizes the optimal sampling overhead required to generate the bipartite state $\rho_{AB}$ using quasiprobabilistic circuit knitting.
The following result asserts that classical communication does not change the sampling overhead for the task of state preparation.
Denote by $\SEP(A,B)$ the set of separable quantum states on $A \otimes B$.
\begin{lemma} \label{lem_CC_not_helpful} 
For any bipartite density operator $\rho_{AB}$ we have
\begin{align*}
    \gamma_{\LO}(\rho_{AB}) = \gamma_{\LOCCOW}(\rho_{AB}) = \gamma_{\LOCC}(\rho_{AB}) \, ,
\end{align*}
where
\begin{align}\label{eq_gamma_state_prep}
    \gamma_{\LOCC}(\rho_{AB}) = \min \Big\{ a_+ + a_- :  \rho_{AB} = a_+\rho_+ - a_-\rho_- , \rho_{\pm}\in\SEP(A,B)\textnormal{ and } a_{\pm}\geq 0 \Big\} \, .
\end{align}
\end{lemma}
The proof is given in~\cref{app_proof_lem_cc_not_helpful}.
It turns out that $\gamma_{\LOCC}(\rho_{AB})$ is an entanglement measure for the state $\rho_{AB}$.
More precisely, it is directly related to a well-studied entanglement monotone called \emph{robustness of entanglement}~\cite{vidal99}, which is defined as
\begin{align*}
    E(\rho_{AB})\coloneqq \min_{\sigma_{AB} \in \SEP} R(\rho_{AB}|| \sigma_{AB}) \, ,
\end{align*}
where $R(\rho_{AB} || \sigma_{AB})\coloneqq\min_{t\geq 0}\{ t : \frac{\rho_{AB}+t \sigma_{AB}}{1+t} \in \mathrm{SEP}\}$.
\begin{lemma} \label{lem_relation_entanglement}
For any density operator $\rho_{AB}$ we have $\gamma_{\LOCC}(\rho_{AB}) = 1 + 2 E(\rho_{AB})$.
\end{lemma}
\begin{proof}
Because $\tr{\rho_{AB}}=1$,~\cref{eq_gamma_state_prep} implies $a_+-a_-=1$.
Therefore we can write
\begin{align*}
    \gamma_{\LOCC}(\rho_{AB}) &= \min \Big\{ 1+2a_- :  \rho_{AB} = (1+a_-)\rho_+ - a_-\rho_- , \rho_{\pm}\in\SEP(A,B)\textnormal{ and } a_{-}\geq 0 \Big\} \\
    &= 1 + 2 \min \Big\{ t : \rho_{AB} = (1+t)\rho_+ - t\rho_- , \rho_{\pm}\in\SEP(A,B)\textnormal{ and } t\geq 0 \Big\} \\
    &= 1 + 2 \min \Big\{ t : \frac{\rho_{AB} + t\rho_-}{1+t} \in \SEP(A,B) , \rho_{-}\in\SEP(A,B)\textnormal{ and } t\geq 0 \Big\} \\
    &= 1 + 2 E(\rho_{AB}) \, .
\end{align*}
\end{proof}
As a consequence from~\cref{lem_relation_entanglement}, $\gamma_{\LOCC}$ inherits many properties from the robustness of entanglement, such as
\begin{itemize}
    \item \emph{Faithfullness:} $\gamma_{\LOCC}(\rho_{AB})=1$ if and only if $\rho_{AB}$ separable.
    \item \emph{Monotonicity:} For any LOCC protocol $\mathcal{E}$: $\gamma_{\LOCC}(\mathcal{E}(\rho_{AB}))\leq \gamma_{\LOCC}(\rho_{AB})$.
    \item \emph{Invariance under local untiary}: $\gamma_{\LOCC}( (U_A\otimes U_B)\rho_{AB} (U_A^{\dagger}\otimes U_B^{\dagger}))=\gamma_{\LOCC}(\rho_{AB})$ for unitaries $U_A,U_B$ acting on $A$ and $B$, respectively.
\end{itemize}

The original motivation for the robustness of entanglement was to measure how much noise one can add to a state before it becomes separable.
One can view~\cref{lem_relation_entanglement} as a new operational interpretation of the entanglement measure: The robustness of entanglement characterizes the simulation overhead of preparing some entangled state with local operations.
For general mixed states $\rho_{AB}$ computing the robustness of entanglement is nontrivial since even determining if $\rho_{AB}$ is entangled (i.e.~$E(\rho_{AB})>0$) or not (i.e.~$E(\rho_{AB})=0$) is known to be NP-hard~\cite{gurvits03}.
Luckily, for pure states detecting entanglement is much simpler and can be done for example via the Schmidt coefficients. Similarly, it has been shown that for pure states there is an explicit expression for the robustness of entanglement, which by~\cref{lem_relation_entanglement} gives us an explicit expression for the $\gamma$-factor of the optimal QPD.
\begin{lemma}[\cite{vidal99}]  \label{lem_vidal} 
Let $\ket{\psi}_{AB}$ be a bipartite state with Schmidt coefficients $\{\alpha_i\}_i$. Then
\begin{align*}
    E(\proj{\psi}) = \left(\sum_i \alpha_i\right)^2 - 1 \, .
\end{align*}
Furthermore, the following QPD 
\begin{align*}
    \proj{\psi} = \big(1+E(\proj{\psi})\big)\rho^{+} - E(\proj{\psi}) \rho ^- \, ,
\end{align*}
where \smash{$\rho^- = \frac{1}{E(\proj{\psi})}\sum_{i \ne j}\alpha_i\alpha_j\proj{ij} \in \SEP$} and \smash{$\rho^+ = \frac{1}{1+E(\proj{\psi})}\left(\proj{\psi} +  E(\proj{\psi})\rho^-\right)$} $\in \SEP $ is optimal. 
\end{lemma}
Note that there is a simple way to compute the Schmidt coefficients $\{\alpha_i\}_i$ of $\ket{\psi}_{AB}$. For $\rho_{AB}=\proj{\psi}_{AB}$, the eigenvalues of $\rho_A$ are exactly given by $\{\alpha^2_i\}_i$.

\begin{corollary}\label{corr_gamma_ebit}
The $\gamma$-factor for the local preparation of $n$ Bell pairs $\{\ket{\Psi_i}_{A_i B_i}\}_{i=1}^n$ for $A=A_1\otimes \ldots \otimes A_n$ and $B=B_1 \otimes \ldots \otimes B_n$ is given by
  \begin{align*}
      \gamma_{\LOCC(A:B)}(\proj{\Psi_1}_{A_1 B_1} \otimes \ldots \otimes \proj{\Psi_n}_{A_n B_n}) = 2^{n+1} - 1 \, .
  \end{align*}
\end{corollary}
\begin{proof}
Note that $n$ Bell pairs have $2^n$ Schmidt coefficients which are all equal to $(1/\sqrt{2})^n$. As a result,~\Cref{lem_vidal} together with~\Cref{lem_relation_entanglement} implies the assertion.
\end{proof}
From~\Cref{corr_gamma_ebit} we see that the optimal sampling overhead is strictly submultiplicative under the tensor product. More precisely, for a Bell state $\ket{\Psi}$ we have
\begin{align*}
    \sqrt{\gamma_{\LOCC(A_1 \otimes A_2:B_1 \otimes B_2)}(\proj{\Psi}_{A_1 B_1} \otimes \proj{\Psi}_{A_2 B_2})} = \sqrt{7} < 3 = \gamma_{\LOCC(A:B)}(\proj{\Psi}_{AB}) \, .
\end{align*}
This shows that locally preparing two Bell states at once is cheaper than individually preparing two Bell states in sequence. The reason is that in the parallel preparation process we can make use of entanglement between $A_1$, $A_2$ and $B_1$, $B_2$, respectively.

%%%%%%%%%%%%%%%%%%%%%%%%%%%%%%%%%%%%%%%%%%%%%%%%%%%%%%%%%%%%%%%%%%%%%%%%%%%%%%%%%%%%%%%
\section{Optimal decompositions for single instances}\label{sec_single_instance}
The goal of this section is to understand the optimal sampling overhead for a single instance of a two-qubit unitary $U$ with local operations with or without classical communication. In technical terms, we want to characterize $\gamma_{\LOCC}(U)$, $\gamma_{\LOCCOW}(U)$ and $\gamma_{\LO}(U)$.
We already know the trivial relation from~\cref{eq_trivial}.
Our main result (see~\cref{thm_main} and~\cref{cor_bounds_match}) is that for a large class of two-qubit unitaries, all three quantities are equal.
We show this result by first finding a lower bound to $\gamma_{\LOCC}(U)$ and compare that to a previously known upper bound~\cite{MF_21} to $\gamma_{\LO}(U)$ and then showing that both bounds coincide for our considered class of two-qubit unitaries.

For a unitary $U_{AB}$ on $A\otimes B$ its Choi state is defined as
\begin{align*} 
  \ket{\Phi_U}_{AA'BB'} \coloneqq \left( U_{AB}\otimes \id_{A'B'}\right) \left( \ket{\Psi}_{AA'}\otimes \ket{\Psi}_{BB'} \right) \, ,
\end{align*}
where $A'$ and $B'$ are identical copies to the systems $A$ and $B$, respectively. The states $\ket{\Psi}_{AA'}$ and $\ket{\Psi}_{BB'}$ are maximally entangled states between $A$ and $A'$, or between $B$ and $B'$ respectively.

\begin{lemma}\label{lem_lower_bound}
  Let $U_{AB}$ be a bipartite unitary with Choi state $\ket{\Phi_U}_{AA'BB'}$ that has Schmidt coefficients $\{\alpha_i\}_i$ when considered as bipartite state over the systems $A\otimes A'$ and $B \otimes B'$. Then,
  \begin{align*}
      \gamma_{\LOCC}(U) \geq \gamma_{\LOCC}(\proj{\Phi_U}) = 2\left(\sum_{i} \alpha_i \right)^2 - 1 \, .
  \end{align*}
\end{lemma}

\begin{proof}
  The desired statement follows from a simple resource-theoretic argument.
  The Choi state $\ket{\Phi_U}$ of $U$ has $\gamma_{\LOCC}(\proj{\Phi_U}) = 2(\sum_{i} \alpha_i)^2 - 1$ by~\cref{lem_vidal}.
  Since the Choi state can be physically realized with an instance of $U$, its  $\gamma$-factor cannot possibly be larger.
\end{proof}
% \begin{figure}[htb!]
%     \centering
%     \begin{tikzpicture}[thick, scale=1]
%     \def\x{4}
%      \draw (0.7,-0.15) -- (1.7,-0.15);
%      \draw (0.7,+0.15) -- (1.7,+0.15);     
%      \draw[fill=white] (1,-0.3) rectangle (1.4,0.3);
%      \draw (0.7,+0.15) -- (0.3,+0.35);   
%      \draw (0.7,+0.55) -- (0.3,+0.35);
%      \draw (0.7,+0.55) -- (1.7,0.55); 
%      \draw (0.7,-0.15) -- (0.3,-0.35);   
%      \draw (0.7,-0.55) -- (0.3,-0.35);
%      \draw (0.7,-0.55) -- (1.7,-0.55);   
%      \node at (-0.35,-0.35) {$\ket{\Psi}_{B_1 B_2}$};
%      \node at (-0.35,+0.35) {$\ket{\Psi}_{A_1 A_2}$};  
%      \node at (1.2,0) {$U$};
%      \node at (2.4,0) {$\ket{\phi}_{A_1 B_1}$};
%     \end{tikzpicture}
%     \caption{Definition of the Choi state $\ket{\phi}_{A_1 B_1}$ for a unitary $U$ acting on $A_1 \otimes B_1$, where $\ket{\Psi}$ denotes a Bell state. }
%     \label{fig_Choi}
% \end{figure}
Any two-qubit unitary has a KAK decomposition
\begin{align} \label{eq_standard_dec}
    U_{AB} = \left(V_1\otimes V_2\right)  \exp\left( \ci\theta_X X \otimes X + \ci\theta_Y Y \otimes Y + \ci\theta_Z Z\otimes Z \right) \left(V_3\otimes V_4\right)  \, ,
\end{align}
for some single-qubit unitaries $V_1,V_2,V_3,V_4$, and $\theta_X,\theta_Y,\theta_Z \in\mathbb{R}$ such that $|\theta_Z| \leq \theta_Y \leq \theta_X \leq \pi/4$.
Evaluating the exponential function, this can be rewritten as
\begin{align}  \label{eq_standard_dec_u}
    U_{AB} = \left(V_1\otimes V_2 \right)  \left(u_0 \id \otimes \id + u_1 X \otimes X + u_2 Y \otimes Y + u_3 Z\otimes Z \right)  \left(V_3\otimes V_4\right) \, ,
\end{align}
where $u_0,u_1,u_2,u_3\in\mathbb{C}$ depend on $\theta_X,\theta_Y,\theta_Z$ and fulfill $\sum_{i=0}^3 |u_i|^2 = 1$.
\begin{lemma}[{\cite{MF_21}}]\label{lem_upper_bound}
  For any two-qubit unitary $U$ we have
  \begin{align*}
      \gamma_{\LO}(U) \leq 1 + \sum_{i\neq j} \left( \lvert u_iu_j^* + u_ju_i^*\rvert + \lvert u_iu_j^* - u_ju_i^*\rvert \right)\, .
  \end{align*}
\end{lemma}

\begin{theorem} \label{thm_main}
  Let $U$ be a two-qubit unitary with parameters $\theta_X,\theta_Y,\theta_Z$ according to the KAK decomposition from~\cref{eq_standard_dec} such that one of the following two conditions holds:
  \begin{enumerate}
      \item $\theta_Z=0$ \label{it_first}
      \item $\theta_X=\theta_Y=\theta_Z=\frac{\pi}{4}$. \label{it_second}
  \end{enumerate}
  Then,
  \begin{align*}
      \gamma_{\LO}(U) 
      &= \gamma_{\LOCCOW}(U)  \\
      &= \gamma_{\LOCC}(U)    \\
      &=1 + 4|\sin\theta_X \cos \theta_X | + 4|\sin\theta_Y \cos \theta_Y | + 8|\sin\theta_X \cos\theta_X \sin\theta_Y \cos\theta_Y|  \\
      &= 2\left(\sum_{i} \alpha_i \right)^2 - 1 \, , 
  \end{align*}
  where $\{\alpha_i\}_i$ are the Schmidt coefficients of the Choi state of $U$ between $A\otimes A'$ and $B\otimes B'$.
\end{theorem}
To prove this, we show that the lower bound from~\cref{lem_lower_bound} and the upper bound from~\cref{lem_upper_bound} coincide, which then implies the assertion by~\cref{eq_trivial}.
The detailed proof is given in~\cref{app_proof_thm_main}.
We give some examples of two-qubit gates of which we can determine the $\gamma$-factor using~\cref{thm_main}.
\begin{corollary}\label{cor_bounds_match}
  For following two-qubit gates
  \begin{itemize}
      \item two-qubit Clifford gates
      \item controlled Rotation gates $\mathrm{CR}_\sigma(\theta)$ for $\sigma\in\{X,Y,Z\}$ and $\theta\in[0,2\pi]$
      \item two-qubit rotations $\mathrm{R}_{\sigma\sigma}(\theta)\coloneqq \exp \left(-\ci \frac{\theta}{2} \sigma\otimes \sigma\right)$ for $\sigma\in\{X,Y,Z\}$ and $\theta\in[0,2\pi]$
  \end{itemize}
  there is no advantage in using classical communication for a single gate instance, i.e.,
\begin{alignat*}{4}
      &\gamma_{\LOCC}(\CNOT)  &&= \gamma_{\LOCCOW}(\CNOT)  &&= \gamma_{\LO}(\CNOT)  &&= 3 \\
      &\gamma_{\LOCC}(\iSWAP)  &&= \gamma_{\LOCCOW}(\iSWAP)  &&= \gamma_{\LO}(\iSWAP)  &&= 7 \\
      &\gamma_{\LOCC}(\SWAP)  &&= \gamma_{\LOCCOW}(\SWAP)  &&= \gamma_{\LO}(\SWAP)  &&= 7 \\
      &\gamma_{\LOCC}\big(\mathrm{CR}_{\sigma}(\theta)\big)  &&= \gamma_{\LOCCOW}\big(\mathrm{CR}_{\sigma}(\theta)\big)  &&= \gamma_{\LO}\big(\mathrm{CR}_{\sigma}(\theta)\big)  &&= 1+2|\sin(\theta /2)| \\
      &\gamma_{\LOCC}\big(\mathrm{R}_{\sigma\sigma}(\theta)\big)  &&= \gamma_{\LOCCOW}\big(\mathrm{R}_{\sigma\sigma}(\theta)\big)  &&= \gamma_{\LO}\big(\mathrm{R}_{\sigma\sigma}(\theta)\big)  &&= 1+2|\sin\theta| \, ,
\end{alignat*}
 for $\sigma \in \{X,Y,Z\}$ and $\theta \in [0,2\pi]$. 
\end{corollary}
Note that every two-qubit Clifford gate is equivalent up to local unitaries to either the identity gate $\id$, the $\CNOT$ gate, the $\iSWAP$ gate or the $\SWAP$ gate.
Therefore, due to~\cref{lem_invariance_local_unitary}, the $\gamma$-factor of any two-qubit Clifford gate is either $1,3,$ or $7$ accordingly.
To prove~\cref{cor_bounds_match}, it suffices to show that the assumption of~\cref{thm_main} is satisfied for the considered gates. The proof is given in~\cref{app_cor_gates}.

For the gates specified by~\cref{thm_main} and~\cref{cor_bounds_match} for which we have an analytical understanding of their optimal sampling overhead, we also explicitly know an optimal QPD.
Since the upper bound from~\cref{lem_upper_bound} is shown to be tight, the QPD introduced in~\cite{MF_21} is optimal.

%%%%%%%%%%%%%%%%%%%%%%%%%%%%%%%%%%%%%%%%%%%%%%%%%%%%%%%%%%%%%%%%%%%%%%%%%%%%%%%%%%%%%%%%%%%%%%%%%%%%%%%%
%%%%%%%%%%%%%%%%%%%%%%%%%%%%%%%%%%%%%%%%%%%%%%%%%%%%%%%%%%%%%%%%%%%%%%%%%%%%%%%%%%%%%%%%%%%%%%%%%%%%%%%%%%%%%%%%%%%%%%
\section{Reducing overhead for multiple instances} \label{sec_how_use_submult}
Under $\LOCC$, a $\CNOT$ gate and a Bell pair denoted by $\ket{\Psi}$ can be considered equally powerful resources: One can generate a Bell pair using a $\CNOT$ gate (see~\cref{fig_BellPair_LOCC_bell}) and one can realize a $\CNOT$ gate using gate teleportation by consuming a preexisting Bell pair (see~\cref{fig_BellPair_LOCC_cnot}).
This implies that under $\LOCC$ it is equally hard to simulate a $\CNOT$ gate than it is to simulate a Bell pair, i.e., $\gamma_{\LOCC}(\CNOT) = \gamma_{\LOCC}(\proj{\Psi})$.
Therefore, we can reduce the problem of quasiprobability simulation of $n$ $\CNOT$ gates dispersed throughout our circuit to the quasiprobability simulation of $n$ Bell pairs at the corresponding locations, without any loss of optimality.
The major advantage of working with states instead of gates is that we can generate the nonlocal resources in parallel and thus make use of the strict subadditivity of the $\gamma$-factor discussed in~\cref{corr_gamma_ebit}.
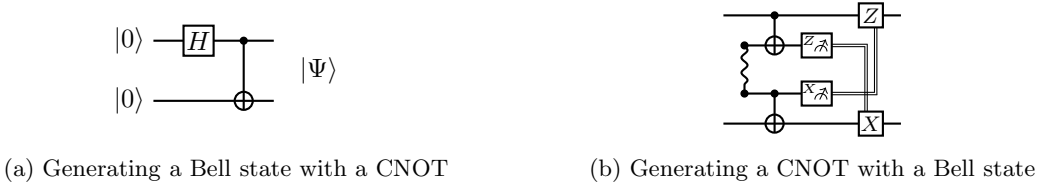
\begin{figure}[!htb]
\centering
\begin{subfigure}{.5\textwidth}
  \centering
    \begin{tikzpicture}[thick,scale=0.8]
    \draw (0,0.5) -- (2,0.5);
    \draw (0,-0.5) -- (2,-0.5);
    \draw[fill=white] (0.5,0.75) rectangle (1.0,0.25); 
    \draw[fill=black] (1.5,0.5) circle (0.5mm); 
    \draw (1.5,-0.5) circle (1.5mm);
    \draw (1.5,0.5) -- (1.5,-0.65);
    \node at (0.75,0.5) {$H$};
    \node at (-0.4,0.5) {$\ket{0}$};
    \node at (-0.4,-0.5) {$\ket{0}$};   
    \node at (2.75,0) {$\ket{\Psi}$};
    \node at (0,-1) {};
    \node at (0,1) {};    
    \end{tikzpicture}
  \caption{Generating a Bell state with a $\CNOT$} \label{fig_BellPair_LOCC_bell}
\end{subfigure}%
\begin{subfigure}{.5\textwidth}
  \centering
    \begin{tikzpicture}[thick,scale=0.8]
    \def\x{0}
    \def\cc{0.02}
    \def \g{0}
%%%%%%%%%%%%%%%%%%%%%%%%%%%%%%%%%%%%%%%%%%%%%%%%%%%%
\draw (\x+\g,0.9) -- (\x+\g+2.25,0.9);
\draw (\g+\x,-0.9) -- (\g+\x+2.25,-0.9);
\draw (\g+\x+2.25+0.4,0.9) -- (\g+\x+2.25+0.7,0.9);
\draw (\g+\x+2.25+0.4,-0.9) -- (\g+\x+2.25+0.7,-0.9);
\draw[fill=black] (\g+\x+0.35,0.4) circle (0.5mm);
\draw[fill=black] (\g+\x+0.35,-0.4) circle (0.5mm);
\draw[decorate,decoration={snake,amplitude=.4mm,segment length=2mm,post length=1mm}] (\g+\x+0.35,0.4) -- (\g+\x+0.35,-0.4);
\draw (\g+\x+0.85,0.4) circle (1.5mm);
\draw[fill=black] (\g+\x+0.85,0.9) circle (0.5mm);
\draw (\g+\x+0.85,0.9) -- (\g+\x+0.85,0.25);
\draw (\g+\x+0.85,-0.9) circle (1.5mm);
\draw[fill=black] (\g+\x+0.85,-0.4) circle (0.5mm);
\draw (\g+\x+0.85,-1.05) -- (\g+\x+0.85,-0.4);
\draw (\g+\x+0.35,0.4) -- (\g+\x+1.3,0.4);
\draw (\g+\x+0.35,-0.4) -- (\g+\x+1.3,-0.4);
\draw (\g+\x+1.3,0.4-0.2) rectangle (\g+\x+1.3+0.5,0.4+0.2);
\draw (\g+\x+1.3,-0.4-0.2) rectangle (\g+\x+1.3+0.5,-0.4+0.2);
\draw (\g+\x+2.25,0.9-0.2) rectangle (\g+\x+2.25+0.4,0.9+0.2);
\draw (\g+\x+2.25,-0.9-0.2) rectangle (\g+\x+2.25+0.4,-0.9+0.2);
\node at (\g+\x+2.25+0.2,0.9) {\footnotesize{$Z$}};
\node at (\g+\x+2.25+0.2,-0.9) {\footnotesize{$X$}};
\draw[thin,->] (\g+\x+1.3+0.25+0.05,0.4-0.1) -- (\g+\x+1.3+0.25+0.1+0.05,0.4+0.1);
\node[] at (\g+\x+1.3+0.12,0.4+0.08) {\tiny{$Z$}};
\draw[thin] (\g+\x+1.3+0.25+0.1+0.07,0.4-0.1) arc (0:180:0.1);
\draw[thin,->] (\g+\x+1.3+0.25+0.05,-0.4-0.1) -- (\g+\x+1.3+0.25+0.1+0.05,-0.4+0.1);
\node[] at (\g+\x+1.3+0.13,-0.4+0.08) {\tiny{$X$}};
\draw[thin] (\g+\x+1.3+0.25+0.1+0.07,-0.4-0.1) arc (0:180:0.1); 
\draw[thin] (\g+\x+1.3+0.5,0.4-\cc) -- (\g+\x+2.25+0.2-\cc-\cc-\cc-\cc-\cc+0.007,0.4-\cc);
\draw[thin] (\g+\x+1.3+0.5,0.4+\cc) -- (\g+\x+2.25+0.2+\cc-\cc-\cc-\cc-\cc+0.007,0.4+\cc);
\draw[thin] (\g+\x+2.25+0.2+\cc-\cc-\cc-\cc-\cc,0.4+\cc) -- (\g+\x+2.25+0.2+\cc-\cc-\cc-\cc-\cc,-0.7);
\draw[thin] (\g+\x+2.25+0.2-\cc-\cc-\cc-\cc-\cc,0.4-\cc) -- (\g+\x+2.25+0.2-\cc-\cc-\cc-\cc-\cc,-0.7);
\draw[thin] (\g+\x+1.3+0.5,-0.4+\cc) -- (\g+\x+2.25+0.2-\cc+\cc+\cc+\cc+\cc+0.007,-0.4+\cc);
\draw[thin] (\g+\x+1.3+0.5,-0.4-\cc) -- (\g+\x+2.25+0.2+\cc+\cc+\cc+\cc+\cc+0.007,-0.4-\cc);
\draw[thin] (\g+\x+2.25+0.2+\cc+\cc+\cc+\cc+\cc,-0.4-\cc) -- (\g+\x+2.25+0.2+\cc+\cc+\cc+\cc+\cc,+0.7);
\draw[thin] (\g+\x+2.25+0.2-\cc+\cc+\cc+\cc+\cc,-0.4+\cc) -- (\g+\x+2.25+0.2-\cc+\cc+\cc+\cc+\cc,+0.7);=
    \end{tikzpicture}
  \caption{Generating a $\CNOT$ with a Bell state} \label{fig_BellPair_LOCC_cnot}
\end{subfigure}
\caption{Depiction why a $\CNOT$ gate and a Bell state are equally powerful resources under $\LOCC$. The wavy line depicts a Bell pair.}
\label{fig_BellPair_LOCC}
\end{figure}

More precisely, one can generate the required Bell pairs all at once and store them until they are consumed via gate teleportation.
This procedure is conceptually depicted in~\cref{fig_CNOT}.
Compared to naively applying the optimal QPD from~\cref{sec_single_instance}, the sampling overhead for simulating a circuit with $n$ nonlocal $\CNOT$ gates can be reduced from $\gamma_{\LOCC}(\CNOT)^{2n}=\gamma_{\LOCC}(\proj{\Psi})^{2n}=9^n$ to \smash{$\gamma^{(n)}_{\LOCC}(\CNOT)^{2n}$} $=\gamma_{\LOCC}(\proj{\Psi}^{\otimes n})^2=(2^{n+1}-1)^2 = \cO(4^n)$ by~\cref{corr_gamma_ebit}.
%This corresponds to an effective $\gamma$-factor of $\gamma_{\LOCC}(\proj{\Psi}^{\otimes n})^{2/n}=(2^{n+1}-1)^{2/n}$ per $\CNOT$ gate.
%In the asymptotic limit $n\rightarrow\infty$, the effective $\gamma$-factor per $\CNOT$ gate is thus given by $\lim\limits_{n\rightarrow\infty}(2^{n+1}-1)^{1/n}=2$.

%\paragraph{Memory overhead}
The technique above comes with the drawback of requiring additional quantum memory to store the Bell states. 
Indeed, if one where to generate all $n$ Bell pairs at the very beginning of the circuit, the size of this additional memory would grow linearly with the number of $n$ nonlocal $\CNOT$ gates.
This is inconvenient, as the limited number of qubits precisely is the motivation for circuit knitting in the first place.
For practical purposes, it is more useful to generate a fixed number $k\leq n$ of Bell pairs at a time and then, once all $k$ Bell pairs have been consumed by gate teleportation, reuse these $2k$ qubits to generate new Bell pairs.
Choosing the size $k$ of this \emph{entanglement factory} comes with a tradeoff: Choosing $k$ smaller results in a reduced memory footprint of the method and increasing $k$ decreases the effective $\gamma$-factor per $\CNOT$ gate, which is given by 
\begin{align} \label{eq_tradeoff}
  \gamma^{(k)}_{\LOCC}(\CNOT)\coloneqq \left( \gamma_{\LOCC}(\CNOT^{\otimes k}) \right)^{1/k} = (2^{k+1}-1)^{1/k} \, .
\end{align}
\cref{fig_tradeoff} graphically visualizes the tradeoff between reducing the effective sampling overhead at the cost of a larger quantum memory.
Note that $\lim_{k\rightarrow\infty} \gamma^{(k)}_{\LOCC}(\CNOT) = 2 < 3 =\gamma_{\LOCC}(\CNOT)$.
The overall sampling overhead for the full circuit consisting of $n$ nonlocal $\CNOT$ gates is \smash{$\gamma^{(k)}_{\LOCC}(\CNOT)^{2n}$}.
\begin{figure}[!htb]
\centering

  \begin{tikzpicture}
	\begin{axis}[
		height=5.0cm,
		width=8.0cm,
		grid=major,
		xlabel=$k$,
%		ylabel=$\beta$,
		xmin=1,
		xmax=20,
		ymax=3,
		ymin=2,
	     xtick={1,5,10,15,20},
	     xticklabels={$1$, $5$, $10$, $15$, $20$},
          ytick={3,2.75,2.5,2.25,2},
		legend style={at={(0.5725,1.002)},anchor=north,legend cell align=left,font=\footnotesize} 
	]

	\addplot[thick,smooth,name path=f] coordinates {
(1,3.0) (2,2.64575) (3,2.46621) (4,2.35961) (5,2.29017) (6,2.24199) (7,2.20694) (8,2.18048) (9,2.15988) (10,2.14344) (11,2.13003) (12,2.1189) (13,2.10952) (14,2.10151) (15,2.09459) (16,2.08855) (17,2.08323) (18,2.07852) (19,2.07431) (20,2.07053)
	};
	\addlegendentry{$\gamma^{(k)}_{\LOCC}(\CNOT)$ see~\cref{eq_tradeoff}}
	
	\end{axis} 
\end{tikzpicture}
\caption{Grpahical visualization of the tradeoff between the entanglement factory size $k$ and the effective sampling overhead for a $\CNOT$ gate \smash{$\gamma^{(k)}_{\LOCC}(\CNOT)$}.}
\label{fig_tradeoff}
\end{figure}
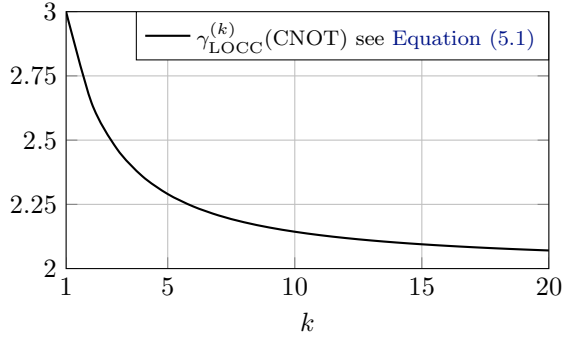

%%%%%%%%%%%%%%%%%%%%%%%%%%%%%%%%%%%%%%%%%%%%%%%%%%%%%%%%%%%%%%%%%%%%%%%%%%%%%%%%%%%%%%%%%%%%%%%%%%%%%%%%
%%%%%%%%%%%%%%%%%%%%%%%%%%%%%%%%%%%%%%%%%%%%%%%%%%%%%%%%%%%%%%%%%%%%%%%%%%%%%%%%%%%%%%%%%%%%%%%%%%%%%%%%
\subsection{General Clifford gates}\label{sec_further_gates}
The main ingredient that enables us to use the submultiplicativity of the $\gamma$-factor under the tensor product is the ability to realize a gate with a $\LOCC$ protocol that consumes a preexisting entangled state.
The gate teleportation protocol can be generalized to realize any bipartite unitary $U_{AB}$ given access to its Choi state $\ket{\Phi_U}_{AA'BB'}$~\cite{GC99}.
Let us denote the system on which the gate is be applied $\tilde{A}\otimes\tilde{B}$.
Both parties perform Bell measurements across $A'$ and $\tilde{A}$ or $B'$ and $\tilde{B}'$ respecitvely.
According to the obtained measurement outcomes, some correction operation has to be applied on $A\otimes B$.
When $U$ is a Clifford gate, this correction operation is an element of the Pauli group and can therefore be realized locally.\footnote{Recall that the Clifford group on $d$ qubits $C_d$ is defined as $C_d=\{ d\textnormal{-qubit unitaries } V : V P_d V^\dagger = P_d \}$, where $P_d$ is the Pauli group on $d$ qubits.}
For example,~\cref{fig_Gate_teleporation} depicts the gate teleportation of a general two-qubit unitary $U$, and the correction operation $U(\sigma_{kl}\otimes\sigma_{ij})U^{\dagger}$ is Pauli when $U$ is Clifford.

\begin{figure}[htb!]
    \centering
    \begin{tikzpicture}[thick,scale=0.8]
    \def\c{0.02}
    % \draw[gray,thin,dashed] (-0.25,0) -- (5,0);
    
     \draw [fill=gray!15,draw=none] (-1,1.5) rectangle (8.3,0.05);
      \draw [fill=blue!15,draw=none] (-1,-1.5) rectangle (8.3,-0.05);  
    
     \draw (-1.2,1.25) -- (4,1.25);  
     \draw (0,0.75) -- (4,0.75);    
     \draw (0,0.25) -- (8.5,0.25);
     \draw (0,-0.25) -- (8.5,-0.25); 
     \draw (0,-0.75) -- (4,-0.75); 
     \draw (-1.2,-1.25) -- (4,-1.25);       
     \draw[fill=white] (0.5,0.5) rectangle (1,1.00); 
     \draw[fill=white] (0.5,-0.0) rectangle (1,-0.5);   
     \node at (0.75,0.75) {$H$};
     \node at (0.75,-0.25) {$H$};      
    \draw[fill=black] (1.5,0.75) circle (0.5mm); 
    \draw (1.5,0.25) circle (1.5mm);
    \draw (1.5,0.75) -- (1.5,0.25-0.14);
    \draw[fill=black] (1.5,-0.25) circle (0.5mm); 
    \draw (1.5,-0.75) circle (1.5mm);    
    \draw (1.5,-0.25) -- (1.5,-0.75-0.14);
    \draw[fill=white] (2,0.5) rectangle (2.5,-0.5);  
    \node at (2.25,0) {$U$};  
    \draw[fill=black] (3.25,1.25) circle (0.5mm);  
    \draw (3.25,0.75) circle (1.5mm); 
    \draw (3.25,1.25) -- (3.25,0.75-0.14);    
    \draw[fill=black] (3.25,-0.75) circle (0.5mm);  
    \draw (3.25,-1.25) circle (1.5mm); 
    \draw (3.25,-0.75) -- (3.25,-1.25-0.14);
    \draw[fill=white] (3.75,1.45) rectangle (4.25,1.05);
    \draw[fill=white] (3.75,0.95) rectangle (4.25,0.55);  
    \draw[fill=white] (3.75,-0.55) rectangle (4.25,-0.95);
    \draw[fill=white] (3.75,-1.05) rectangle (4.25,-1.45);  
    \draw[thin] (4.15,1.15) arc (0:180:0.15);
    \draw[thin,->] (4,1.15) -- (4.15,1.4);
    \draw[thin] (4.15,1.15-0.5) arc (0:180:0.15);
    \draw[thin,->] (4,1.15-0.5) -- (4.15,1.4-0.5); 
    \draw[thin] (4.15,-1.35) arc (0:180:0.15);
    \draw[thin,->] (4,-1.35) -- (4.15,-1.1); 
    \draw[thin] (4.15,-1.35+0.5) arc (0:180:0.15);
    \draw[thin,->] (4,-1.35+0.5) -- (4.15,-1.1+0.5);     
    \draw[fill=white] (4.75,0.5) rectangle (8,-0.5);  
    \node at (12.75/2,0) {$U ( \sigma_{k\ell} \otimes \sigma_{ij}) U^\dagger$};
    
    \draw[thin] (4.25,1.25+\c) -- (4.25+1+\c+0.0088,1.25+\c);
    \draw[thin] (4.25,1.25-\c) -- (4.25+1-\c+0.0088,1.25-\c);
    \draw[thin] (4.25+1+\c,1.25+\c) -- (4.25+1+\c,0.5); 
    \draw[thin] (4.25+1-\c,1.25-\c) -- (4.25+1-\c,0.5); 
    \node at (4.25+0.35,1.45) {\footnotesize{$i$}};

    \draw[thin] (4.25,0.75+\c) -- (4.25+0.7+\c+0.0088,0.75+\c);
    \draw[thin] (4.25,0.75-\c) -- (4.25+0.7-\c+0.0088,0.75-\c);
    \draw[thin] (4.25+0.7+\c,0.75+\c) -- (4.25+0.7+\c,0.5); 
    \draw[thin] (4.25+0.7-\c,0.75-\c) -- (4.25+0.7-\c,0.5); 
    \node at (4.25+0.35,0.97) {\footnotesize{$j$}};   

    \draw[thin] (4.25,-0.75-\c) -- (4.25+0.7+\c+0.0088,-0.75-\c);
    \draw[thin] (4.25,-0.75+\c) -- (4.25+0.7-\c+0.0088,-0.75+\c);
    \draw[thin] (4.25+0.7+\c,-0.75-\c) -- (4.25+0.7+\c,-0.5); 
    \draw[thin] (4.25+0.7-\c,-0.75+\c) -- (4.25+0.7-\c,-0.5); 
    \node at (4.25+0.35,-0.55) {\footnotesize{$k$}};  
 
    \draw[thin] (4.25,-1.25-\c) -- (4.25+1+\c+0.0088,-1.25-\c);
    \draw[thin] (4.25,-1.25+\c) -- (4.25+1-\c+0.0088,-1.25+\c);
    \draw[thin] (4.25+1+\c,-1.25-\c) -- (4.25+1+\c,-0.5); 
    \draw[thin] (4.25+1-\c,-1.25+\c) -- (4.25+1-\c,-0.5); 
    \node at (4.25+0.35,-1.05) {\footnotesize{$\ell$}};
    
    \node at (-0.3,0.75) {$\ket{0}$};
    \node at (-0.3,0.25) {$\ket{0}$};
    \node at (-0.3,-0.25) {$\ket{0}$};
    \node at (-0.3,-0.75) {$\ket{0}$}; 
    \draw[dotted,black] (-0.70,1.1) rectangle (2.7,-1.1); 
    
    \node at (-1.5,0) {$\ket{\psi}$};
    \node at (9.2,0) {$U \ket{\psi}$};    
    
     \draw [fill=gray!15,draw=none] (-1,1+0.7) rectangle (-1+0.25,1+0.25+0.7);
     \draw [fill=blue!15,draw=none] (0.1-1.2+1.75,1+0.7) rectangle (0.1+0.25-1.2+1.75,1+0.25+0.7);     
     \node[gray] at (-1+0.25+0.3,1+0.7+0.125) {\footnotesize{$\bar A$}};
     \node[blue] at (0.1+0.25-1.2+0.3+1.75,1+0.7+0.125) {\footnotesize{$\bar B$}};
    
    \end{tikzpicture}
    \caption{Graphical explanation of a gate teleportation protocol for a two-qubit gate $U$. We realize the gate $U$ with the help of its Choi state (prepared by the dotted box) and a correction operator $U(\sigma_{ij}\otimes\sigma_{kl})U^{\dagger}$, where $\sigma_{ij}\coloneqq Z^iX^j$. For Clifford gates $U$ the correction operator is local and hence we can simulate $U$ with the Choi state and $\LOCC$.}
    \label{fig_Gate_teleporation}
\end{figure}
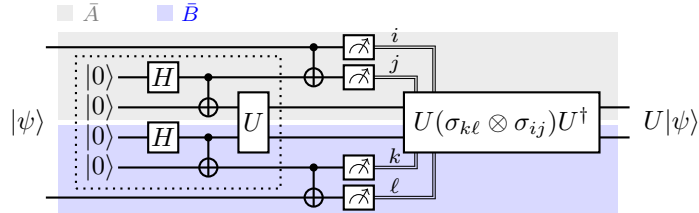

By using this gate teleportation protocol, the technique for reducing the overhead for multiple instances of a $\CNOT$ gate can be straightforwardly generalized to reduce the sampling arbitrary Clifford gates.
We summarize the result of this section:

\begin{theorem} \label{thm_clifford}
  Let $U_{AB}$ be a Clifford unitary with a Choi state $\ket{\Phi_U}_{AA'BB'}$ that has Schmidt coefficients $\{\alpha_i \}_i$ when considered as bipartite state over the systems $A\otimes A'$ and $B \otimes  B'$. Then
  \begin{align} \label{eq_thm1}
    \gamma_{\LOCC}(U) = \gamma_{\LOCC}(\proj{\Phi_U}) = 2 \left(\sum_i \alpha_i \right)^2 -1 \, .
  \end{align}
 For a memory overhead of $2k$ qubits for $k \in \N$, our method can achieve an effective $\gamma$-factor of
  \begin{align}\label{eq_thm2}
      \gamma^{(k)}_{\LOCC}(U) =\left( \gamma_{\LOCC}(\proj{\Phi_U}^{\otimes k}) \right)^{1/k} = \left(2 \Big(\sum_i \alpha_i\Big)^{2k}-1 \right)^{1/k} \, .
  \end{align}
\end{theorem}

\begin{proof}
Under $\LOCC$ we can prepare the state $\ket{\Phi_U}$ using an instance of the gate $U$, and similarly we can realize $U$ using gate teleportation, which consumes the state $\ket{\Phi_U}$.
This implies that $\gamma_{\LOCC}(U)=\gamma_{\LOCC}(\proj{\Phi_U})$.
\Cref{eq_thm1} directly follows from~\cref{lem_relation_entanglement} and~\cref{lem_vidal}.

Following the idea sketched in the main text, we can reduce the simulation of $k$ nonlocal instances of $U$ to the preparation of $\ket{\Phi_U}^{\otimes k}$.
Denote the Schmidt coefficients of $\ket{\Phi_U}$ by $\alpha_1,\dots,\alpha_m$.
Then the Schmidt coefficients of $\ket{\Phi_U}^{\otimes k}$ are given by $\tilde \alpha_{j_1,\dots,j_k}=\prod_{i=1}^k \alpha_{j_i}$ where $j_i$ ranges over $\{1,\dots,m\}$.
By invoking~\cref{lem_vidal}, we obtain
\begin{align*}
    \gamma_{\LOCC}(\proj{\Phi_U}^{\otimes k}) 
    =2\Big(\sum_{j_1,\dots,j_k} \tilde \alpha_{j_1,\dots,j_k} \Big)^{2} -1 
    =2 \Big(\sum_i \alpha_i \Big)^{2k} -1  \, ,
\end{align*}
which gives us the effective $\gamma$-factor described in~\cref{eq_thm2}.
\end{proof}
Entangling Clifford unitaries $U$ satisfy $(\sum_i \alpha_i)^2>1$. For such unitaries~\cref{thm_clifford} implies that\footnote{This can be verified by induction.}
\begin{align*}
  \gamma^{(k)}_{\LOCC}(U) <  \gamma_{\LOCC}(U) \quad \textnormal{for} \quad k>1 \, .  
\end{align*}
This shows that the presented protocol via gate teleportation indeed reduces the sampling overhead for multiple instances of the nonlocal gate $U$.
For a circuit with $n$ instances of $U$ the total sampling overhead is then given as \smash{$\gamma^{(n)}_{\LOCC}(U)^{2n}=O((\sum_i \alpha_i)^{4n})$}, which recovers the $O(4^n)$ and the $O(16^n)$ scaling for the case of $\CNOT$ and $\SWAP$ gates, respectively.\footnote{Recall that the Schmidt coefficients for a $\CNOT$ gate are $\alpha_1=\alpha_2=1/\sqrt{2}$ and for a $\SWAP$ gate $\alpha_1=\alpha_2=\alpha_3=\alpha_4=1/2$.}

%%%%%%%%%%%%%%%%%%%%%%%%%%%%%%%%%%%%%%%%%%%%%%%%%%%%
\subsection{Non-Clifford gates} \label{sec_nonClifford_multiple}
For non-Clifford gates $U$, the simple correspondence under $\LOCC$ between the gate and its Choi state via the gate teleportation circuit from~\cref{fig_Gate_teleporation} does not hold anymore, because the correction operator in the gate teleportation protocol is no longer local.
As a result, a different $\LOCC$ protocol would need to be used in order to realize the gate using a preexisting entangled state.
In the following we consider the $\mathrm{CR}_{X}(\theta)$ gate as a case study of a non-Clifford gate and we show that our method can reduce its sampling overhead at least for certain values of $\theta$.
We want to characterize its lowest achievable effective $\gamma$-factor (i.e. the effective sampling overhead per gate) in the limit where the number of gates goes to infinity.
Note that this discussion analogously holds for all gates that are equivalent to the $\mathrm{CR}_X$ gate up to local unitaries, such as $\mathrm{CR}_Y,\mathrm{CR}_Z,\mathrm{R}_{XX},\mathrm{R}_{YY}$ and $\mathrm{R}_{ZZ}$ (see~\cref{eq_equivalence_crx_rxx}).

As we have seen in~\cref{cor_bounds_match}, one can achieve a $\gamma$-factor of $1+2|\sin(\theta/2)|$ in the $\LO$ setting, so this represents an upper bound to the best possible effective $\gamma$-factor that we can hope to achieve.
By an argument analogous to the proof of~\cref{lem_lower_bound}, the asymptotic effective $\gamma$-factor of a $\mathrm{CR}_{X}(\theta)$ gate cannot be lower than $\lim_{n\rightarrow\infty}\gamma_{\LOCC}(\proj{\Phi_{\mathrm{CR}_{X}(\theta)}}^{\otimes n})^{1/n}=1+|\sin(\theta/2)|$:
The value $1+|\sin(\theta/2)|$ is the lowest possible effective $\gamma$-factor with which a $\ket{\Phi_{\mathrm{CR}_{X}(\theta)}}$ can be prepared, and since a $\mathrm{CR}_{X}(\theta)$ gate can be used to realize a $\ket{\Phi_{\mathrm{CR}_{X}(\theta)}}$ state, its lowest effective $\gamma$-factor cannot possibly be lower.

With this in mind, we know that the lowest achievable effective $\gamma$-factor lies between $1+|\sin(\theta/2)|$ and $1+2|\sin(\theta/2)|$, as seen in~\cref{fig_openQuestion}.
It has been shown~\cite[Theorem 2]{EPP00}, that any controlled rotation gate can be realized with some $\LOCC$ protocol which requires a single Bell pair.
Since the effective $\gamma$-factor of creating a Bell pair is $2$, we can therefore upper bound the lowest achievable effective $\gamma$-factor of the $\mathrm{CR}_{X}(\theta)$ gate by two and it must therefore lie somewhere in the gray region of~\cref{fig_openQuestion}.
Notice that this implies that classical communication can reduce the sampling overhead of the $\mathrm{CR}_{X}(\theta)$ gate for $\pi/3<\theta<5\pi/3$.
For the remaining values of $\theta$, it is an open question whether classical communication can reduce the sampling overhead. 
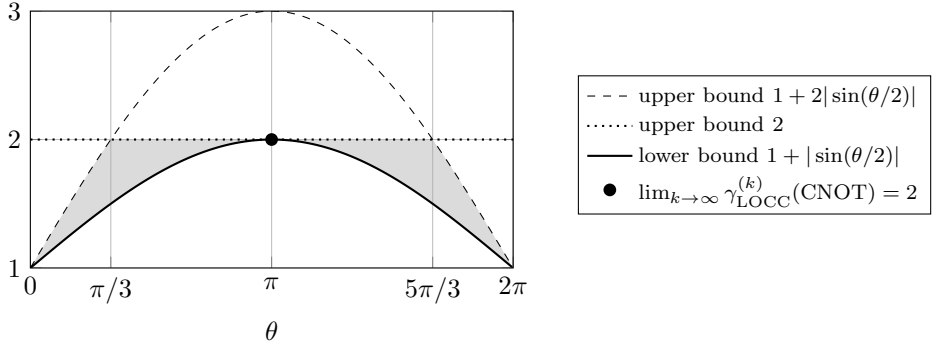
\begin{figure}[!htb]
\centering

  \begin{tikzpicture}
	\begin{axis}[
		height=5.0cm,
		width=8.0cm,
		grid=major,
		xlabel=$\theta$,
%		ylabel=$\beta$,
		xmin=0,
		xmax=6.28319,
		ymax=3,
		ymin=1,
	     xtick={0,1.0472,3.14,5.23599,6.28319},
	     xticklabels={$0$, $\pi/3$, $\pi$, $5\pi/3$, $2\pi$},
          ytick={1,2,3},
		legend style={at={(1.5,0.75)},anchor=north,legend cell align=left,font=\footnotesize} 
	]

	\addplot[dashed,smooth,name path=f] coordinates {
(0.,1.) (0.03,1.03) (0.06,1.05999) (0.09,1.08997) (0.12,1.11993) (0.15,1.14986) (0.18,1.17976) (0.21,1.20961) (0.24,1.23942) (0.27,1.26918) (0.3,1.29888) (0.33,1.3285) (0.36,1.35806) (0.39,1.38753) (0.42,1.41692) (0.45,1.44621) (0.48,1.47541) (0.51,1.50449) (0.54,1.53346) (0.57,1.56231) (0.6,1.59104) (0.63,1.61963) (0.66,1.64809) (0.69,1.67639) (0.72,1.70455) (0.75,1.73255) (0.78,1.76038) (0.81,1.78804) (0.84,1.81552) (0.87,1.84282) (0.9,1.86993) (0.93,1.89685) (0.96,1.92356) (0.99,1.95006) (1.02,1.97635) (1.05,2.00243) (1.08,2.02827) (1.11,2.05389) (1.14,2.07926) (1.17,2.1044) (1.2,2.12928) (1.23,2.15392) (1.26,2.17829) (1.29,2.2024) (1.32,2.22623) (1.35,2.24979) (1.38,2.27307) (1.41,2.29607) (1.44,2.31877) (1.47,2.34117) (1.5,2.36328) (1.53,2.38507) (1.56,2.40656) (1.59,2.42773) (1.62,2.44857) (1.65,2.4691) (1.68,2.48929) (1.71,2.50914) (1.74,2.52866) (1.77,2.54783) (1.8,2.56665) (1.83,2.58513) (1.86,2.60324) (1.89,2.62099) (1.92,2.63838) (1.95,2.6554) (1.98,2.67205) (2.01,2.68832) (2.04,2.70422) (2.07,2.71972) (2.1,2.73485) (2.13,2.74958) (2.16,2.76392) (2.19,2.77786) (2.22,2.7914) (2.25,2.80454) (2.28,2.81727) (2.31,2.82959) (2.34,2.8415) (2.37,2.853) (2.4,2.86408) (2.43,2.87474) (2.46,2.88498) (2.49,2.89479) (2.52,2.90418) (2.55,2.91314) (2.58,2.92167) (2.61,2.92977) (2.64,2.93743) (2.67,2.94466) (2.7,2.95145) (2.73,2.9578) (2.76,2.96371) (2.79,2.96918) (2.82,2.9742) (2.85,2.97878) (2.88,2.98292) (2.91,2.98661) (2.94,2.98985) (2.97,2.99264) (3.,2.99499) (3.03,2.99689) (3.06,2.99834) (3.09,2.99933) (3.12,2.99988) (3.15,2.99998) (3.18,2.99963) (3.21,2.99883) (3.24,2.99758) (3.27,2.99588) (3.3,2.99373) (3.33,2.99113) (3.36,2.98809) (3.39,2.98459) (3.42,2.98065) (3.45,2.97627) (3.48,2.97144) (3.51,2.96616) (3.54,2.96045) (3.57,2.95429) (3.6,2.9477) (3.63,2.94066) (3.66,2.93319) (3.69,2.92528) (3.72,2.91694) (3.75,2.90817) (3.78,2.89897) (3.81,2.88934) (3.84,2.87929) (3.87,2.86882) (3.9,2.85792) (3.93,2.84661) (3.96,2.83488) (3.99,2.82273) (4.02,2.81018) (4.05,2.79722) (4.08,2.78386) (4.11,2.77009) (4.14,2.75593) (4.17,2.74137) (4.2,2.72642) (4.23,2.71108) (4.26,2.69536) (4.29,2.67925) (4.32,2.66277) (4.35,2.64591) (4.38,2.62868) (4.41,2.61109) (4.44,2.59313) (4.47,2.57482) (4.5,2.55615) (4.53,2.53713) (4.56,2.51776) (4.59,2.49805) (4.62,2.47801) (4.65,2.45763) (4.68,2.43693) (4.71,2.4159) (4.74,2.39456) (4.77,2.3729) (4.8,2.35093) (4.83,2.32865) (4.86,2.30608) (4.89,2.28322) (4.92,2.26006) (4.95,2.23662) (4.98,2.21291) (5.01,2.18892) (5.04,2.16466) (5.07,2.14014) (5.1,2.11537) (5.13,2.09034) (5.16,2.06507) (5.19,2.03956) (5.22,2.01381) (5.25,1.98784) (5.28,1.96165) (5.31,1.93523) (5.34,1.90861) (5.37,1.88178) (5.4,1.85476) (5.43,1.82754) (5.46,1.80014) (5.49,1.77256) (5.52,1.7448) (5.55,1.71687) (5.58,1.68879) (5.61,1.66055) (5.64,1.63216) (5.67,1.60362) (5.7,1.57496) (5.73,1.54616) (5.76,1.51724) (5.79,1.4882) (5.82,1.45906) (5.85,1.42981) (5.88,1.40046) (5.91,1.37102) (5.94,1.3415) (5.97,1.31191) (6.,1.28224) (6.03,1.25251) (6.06,1.22272) (6.09,1.19289) (6.12,1.163) (6.15,1.13309) (6.18,1.10314) (6.21,1.07317) (6.24,1.04318) (6.27,1.01319)
	};
	\addlegendentry{upper bound $1+2|\sin(\theta/2)|$}

		\addplot[dotted,thick,smooth,name path=h] coordinates {
    (0.0000,    2)
    (1.57079632679,    2)
    (5.23599,2)
    (6.28319,    2)
	};
		\addlegendentry{upper bound $2$}
	
		\addplot[thick,smooth,name path=g] coordinates {
(0.,1.) (0.03,1.015) (0.06,1.03) (0.09,1.04498) (0.12,1.05996) (0.15,1.07493) (0.18,1.08988) (0.21,1.10481) (0.24,1.11971) (0.27,1.13459) (0.3,1.14944) (0.33,1.16425) (0.36,1.17903) (0.39,1.19377) (0.42,1.20846) (0.45,1.22311) (0.48,1.2377) (0.51,1.25225) (0.54,1.26673) (0.57,1.28116) (0.6,1.29552) (0.63,1.30982) (0.66,1.32404) (0.69,1.3382) (0.72,1.35227) (0.75,1.36627) (0.78,1.38019) (0.81,1.39402) (0.84,1.40776) (0.87,1.42141) (0.9,1.43497) (0.93,1.44842) (0.96,1.46178) (0.99,1.47503) (1.02,1.48818) (1.05,1.50121) (1.08,1.51414) (1.11,1.52694) (1.14,1.53963) (1.17,1.5522) (1.2,1.56464) (1.23,1.57696) (1.26,1.58914) (1.29,1.6012) (1.32,1.61312) (1.35,1.6249) (1.38,1.63654) (1.41,1.64803) (1.44,1.65938) (1.47,1.67059) (1.5,1.68164) (1.53,1.69254) (1.56,1.70328) (1.59,1.71386) (1.62,1.72429) (1.65,1.73455) (1.68,1.74464) (1.71,1.75457) (1.74,1.76433) (1.77,1.77391) (1.8,1.78333) (1.83,1.79256) (1.86,1.80162) (1.89,1.8105) (1.92,1.81919) (1.95,1.8277) (1.98,1.83603) (2.01,1.84416) (2.04,1.85211) (2.07,1.85986) (2.1,1.86742) (2.13,1.87479) (2.16,1.88196) (2.19,1.88893) (2.22,1.8957) (2.25,1.90227) (2.28,1.90863) (2.31,1.91479) (2.34,1.92075) (2.37,1.9265) (2.4,1.93204) (2.43,1.93737) (2.46,1.94249) (2.49,1.9474) (2.52,1.95209) (2.55,1.95657) (2.58,1.96084) (2.61,1.96488) (2.64,1.96872) (2.67,1.97233) (2.7,1.97572) (2.73,1.9789) (2.76,1.98185) (2.79,1.98459) (2.82,1.9871) (2.85,1.98939) (2.88,1.99146) (2.91,1.9933) (2.94,1.99492) (2.97,1.99632) (3.,1.99749) (3.03,1.99844) (3.06,1.99917) (3.09,1.99967) (3.12,1.99994) (3.15,1.99999) (3.18,1.99982) (3.21,1.99942) (3.24,1.99879) (3.27,1.99794) (3.3,1.99687) (3.33,1.99557) (3.36,1.99404) (3.39,1.9923) (3.42,1.99033) (3.45,1.98813) (3.48,1.98572) (3.51,1.98308) (3.54,1.98022) (3.57,1.97715) (3.6,1.97385) (3.63,1.97033) (3.66,1.96659) (3.69,1.96264) (3.72,1.95847) (3.75,1.95409) (3.78,1.94949) (3.81,1.94467) (3.84,1.93965) (3.87,1.93441) (3.9,1.92896) (3.93,1.9233) (3.96,1.91744) (3.99,1.91137) (4.02,1.90509) (4.05,1.89861) (4.08,1.89193) (4.11,1.88505) (4.14,1.87796) (4.17,1.87068) (4.2,1.86321) (4.23,1.85554) (4.26,1.84768) (4.29,1.83963) (4.32,1.83138) (4.35,1.82295) (4.38,1.81434) (4.41,1.80554) (4.44,1.79657) (4.47,1.78741) (4.5,1.77807) (4.53,1.76856) (4.56,1.75888) (4.59,1.74903) (4.62,1.73901) (4.65,1.72882) (4.68,1.71846) (4.71,1.70795) (4.74,1.69728) (4.77,1.68645) (4.8,1.67546) (4.83,1.66433) (4.86,1.65304) (4.89,1.64161) (4.92,1.63003) (4.95,1.61831) (4.98,1.60645) (5.01,1.59446) (5.04,1.58233) (5.07,1.57007) (5.1,1.55768) (5.13,1.54517) (5.16,1.53253) (5.19,1.51978) (5.22,1.50691) (5.25,1.49392) (5.28,1.48082) (5.31,1.46762) (5.34,1.45431) (5.37,1.44089) (5.4,1.42738) (5.43,1.41377) (5.46,1.40007) (5.49,1.38628) (5.52,1.3724) (5.55,1.35844) (5.58,1.34439) (5.61,1.33027) (5.64,1.31608) (5.67,1.30181) (5.7,1.28748) (5.73,1.27308) (5.76,1.25862) (5.79,1.2441) (5.82,1.22953) (5.85,1.2149) (5.88,1.20023) (5.91,1.18551) (5.94,1.17075) (5.97,1.15595) (6.,1.14112) (6.03,1.12625) (6.06,1.11136) (6.09,1.09644) (6.12,1.0815) (6.15,1.06654) (6.18,1.05157) (6.21,1.03658) (6.24,1.02159) (6.27,1.00659)		
	};
		\addlegendentry{lower bound $1+|\sin(\theta/2)|$}

		\addplot[thick,only marks] coordinates {
    (3.14,    2)
	};
		\addlegendentry{$\lim_{k\to \infty }\gamma^{(k)}_{\LOCC}(\CNOT)=2$}
		
\addplot[gray!70, opacity=0.4] fill between[of=f and g, soft clip={domain=0:1.0472}];
\addplot[gray!70, opacity=0.4] fill between[of=h and g, soft clip={domain=1.0472:3.14159265358979323}];
\addplot[gray!70, opacity=0.4] fill between[of=g and h, soft clip={domain=3.14159265358979323:5.23599}];
\addplot[gray!70, opacity=0.4] fill between[of=f and g, soft clip={domain=5.23599:6.28319}];
	\end{axis} 
\end{tikzpicture}
\caption{Known upper and lower bounds for the optimal effective $\gamma$-factor of a $\mathrm{CR}_{X}(\theta)$-gate under $\LOCC$ in the limit where the number of gates goes to infinity. The precise value is only known for $\theta\in\{0,2\pi\}$ as well as $\theta=\pi$ where the gate becomes a $\CNOT$. In the gray area the currently best known upper and lower bounds do not match.}
\label{fig_openQuestion}
\end{figure}

\subsection{One-way classical communication}
It turns out that our technique can also be applied in the $\LOCCOW$ setting, but only with a smaller reduction of sampling overhead.
The gate teleportation in~\cref{fig_CNOT} can be simulated in $\LOCCOW$ by post-selecting the measurement outcome on circuit $B$ to have the outcome $0$, allowing us to forgo the classical communication from $B$ to $A$.
The postselection step increases the overhead by a factor of two, since a measurement outcome $0$ only happens with probability $1/2$.
As a result the total sampling overhead for simulating a circuit with $n$ nonlocal $\CNOT$ gates under \smash{$\LOCCOW$} scales as $O(8^n)$, which is still an improvement over the $O(9^n)$ under $\LO$.

\paragraph{Acknowledgements}
We thank Sergey Bravyi, Julien Gacon, Jay Gambetta, and Stefan Woerner for helpful discussions.
This work was supported by the Swiss National Science Foundation, through the National Center of Competence in Research ``Quantum Science and Technology'' (QSIT) and through grant number 20QT21\_187724.

%%%%%%%%%%%%%%%%%%%%%%%%%%%%%%%%%%%%%%%%%%%%%%%%%%%%%%%%%%%%%%%%%%%%%%%%%%%%%%%%%%%%%%%%%%%%%%%%%%%%%%%%%%%%%%%%%%%%%%

\appendix

\section{Non-positive superoperators in quasiprobability simulation}\label{app_nonpositive_ops}
In this section we briefly outline why we allow for certain non-trace preserving and non-positive operations in our quasiprobability decompositions, i.e. why we define $\LO(A,B)$ as 
\begin{align}\label{eq_lodef_naive}
  \{A\otimes B | \mathcal{A}\in\mathrm{D}(A), \mathcal{B}\in\mathrm{D}(B)\}
\end{align}
instead of 
\begin{align}\label{eq:lodef_refined}
  \{A\otimes B | \mathcal{A}\in\TPCP(A), \mathcal{B}\in\TPCP(B)\} \,.
\end{align}

It was realized in~\cite{endo18} that it can be useful (or often even necessary) to include trace-nonincreasing maps in quasiprobability decompositions.
This can be done because any trace-nonincreasing map can be effectively simulated by some measurement process and post-selection of the corresponding measurement outcome.
More precisely, any map $\mathcal{E}\in\TNCP(A)$ can be extended to a trace-preserving map: $\exists \mathcal{F}\in\TNCP(A)$ s.t. $\mathcal{E}+\mathcal{F}\in\TPCP(A)$.
To simulate $\mathcal{E}$, one can perform the trace-preserving completely positive map
\begin{align}
    \rho_A \mapsto \mathcal{E}(\rho)_A\otimes\proj{0}_E + \mathcal{F}(\rho)_A\otimes\proj{1}_E \, ,
\end{align}
where $E$ is a qubit system.
One then measures $E$ in the computational basis and postselects for the outcome $0$.
In practice, this is done by multiplying the final outcome of the circuit by $0$ in case the measurement outcome $1$ is obtained.

This trick was later generalized by Mitarai and Fujii~\cite{Mitarai_2021,MF_21} who realized that one can even consider a non-completely positive map $\mathcal{E}$, as long as it can be written as a difference of completely-positive trace-nonincreasing maps that add up to another completely-positive trace-nonincreasing map, i.e. $\mathcal{E}=\mathcal{E}^+-\mathcal{E}^-$, $\mathcal{E}^{\pm}\in\TNCP$ and $\mathcal{E}^++\mathcal{E}^-\in\TNCP$.
The idea here is very similar: any such map can be simulated using the trace-nonincreasing completely positive map
\begin{align}
    \rho_A \mapsto \mathcal{E}^+(\rho)_A\otimes\proj{0}_E + \mathcal{E}^-(\rho)_A\otimes\proj{1}_E \, ,
\end{align}
measuring the qubit $E$ in the computational basis and and correspondingly weighting the final measurement outcome of the circuit by $+1$ or $-1$ depending on the outcome.

\section{The $\gamma$-factor is well-defined}\label{app_min_achieved}
In this appendix we briefly argue why the minimum in~\cref{eq_def_gamma} is achieved and therefore why the $\gamma$-factor is well-defined.
Assume that the set $S$ is compact.
This implies that the convex hull
\begin{align*}
    \mathrm{conv}(S) \coloneqq \Big\{ \sum_{i} p_i\mathcal{F}_i : \mathcal{F}_i\in S, p_i\geq0, \sum_i p_i=1 \Big\}
\end{align*}
is also compact.
Consider the quantity
\begin{align*}
    \tilde{\gamma}(\mathcal{E})\coloneqq \min \Big\{ a_+ + a_- : \mathcal{E}=a_+\mathcal{F}_+ - a_-\mathcal{F}_- \text{ where } \mathcal{F}_{\pm}\in \mathrm{conv}(S) \textnormal{ and } a_{\pm}\geq 0 \Big\} \, .
\end{align*}
The expression $\tilde{\gamma}(\mathcal{E})$ is well-defined, since it optimizes a continuous function over a compact set.\footnote{We do not have to consider $a_{\pm}$ to be arbitrarily large, as we know that for any $\mathcal{E}\in\mathscr{S}(A\otimes B)$ there exists a QPD with finite sampling overhead~\cite{endo18}.}
Any QPD of the form $\mathcal{E}=a_+\mathcal{F}_+ - a_-\mathcal{F}_-$, $\mathcal{F}_{\pm}\in \mathrm{conv}(S)$ can be brought to the form $\mathcal{E} = \sum_i a_i \mathcal{F}_i$, $\mathcal{F}_i\in S$ while retaining the identical sampling overhead by definition of the convex hull.
Similarily, any QPD of the second form can be brought to the first form by defining $a_+\coloneqq\sum_{i : a_i\geq 0}a_i$, $a_-\coloneqq\sum_{i : a_i< 0}a_i$, $\mathcal{F}_+\coloneqq\frac{1}{a_+}\sum_{i : a_i\geq 0}\mathcal{F}_i$ and $\mathcal{F}_-\coloneqq\frac{1}{a_-}\sum_{i : a_i< 0}\mathcal{F}_i$ while retaining the same sampling overhead.
Therefore, the well-definedness of $\tilde{\gamma}(\mathcal{E})$ implies that the minimum in~\cref{eq_def_gamma} is achieved.

Note that the conventional definition of the set of $\LOCC$ protocols is actually not a closed set~\cite{CLMOW14}.
Instead, for our purposes, we choose $\LOCC(A,B)$ to be the set of $\LOCC$ protocols with bounded number of rounds, where the bound can be chosen to be arbitrary large.
It has been shown that finite round protocols do form a compact set~\cite{CLMOW14}.
This minor technicality is of little significance for our work, since it turns out that the optimizer of~\cref{eq_def_gamma} is often a $\LOCC$ protocol with few rounds.
%%%%%%%%%%%%%%%%%%%%%%%%%%%%%%%%%%%%%%%%%%%%%%%%%%%%%%%%%%%%%%%%%%%%%%%%%%%%%%%%%%%%%%%%%%%%%%%%%%%%%%%%%%
\section{Proofs}
To improve readability of the manuscript we shifted some proofs to this appendix.
\subsection{Proof of~\cref{lem_CC_not_helpful}} \label{app_proof_lem_cc_not_helpful}
Denote the right-hand side of~\cref{eq_gamma_state_prep} by $\gamma_{\SEP}(\rho_{AB})$.
We only have to show that $\gamma_{\SEP}(\rho_{AB})\leq \gamma_{\LOCC}(\rho_{AB})$ and $\gamma_{\LO}(\rho_{AB})\leq \gamma_{\SEP}(\rho_{AB})$.
The desired statement directly follows directly from~\cref{eq_trivial}.

Consider $\mathcal{E}\in\mathscr{S}(A\otimes B)$ to be the map that minimizes~\cref{eq_def_state_qpd} and let $\mathcal{E}=\sum_i a_i\mathcal{F}_i$ be the quasiprobability decomposition that achieves the lowest sampling overhead in the definition of $\gamma_{\LOCC}(\mathcal{E})$.
The map $\mathcal{F}_i$ can be written as the difference $\mathcal{F}_i=\mathcal{F}_{i,+}-\mathcal{F}_{i,-}$ of two completely-positive trace-nonincreasing $\LOCC$ maps $\mathcal{F}_{i,\pm}$.
It is a well-known fact that the set of states achievable under such maps, without any given entanglement, is precisely the set of separable states and hence $\sigma_{i,\pm} \coloneqq \mathcal{F}_{i,\pm}(\proj{0}_{AB})$ are separable states with $\tr{\sigma_{i,+}} + \tr{\sigma_{i,-}}\leq 1$.

This allows us to write
\begin{align*}
  \rho = \sum\limits_i a_i\sigma_{i,+} - a_i\sigma_{i,-} 
  =  a_+\sigma_+ - a_-\sigma_- \, ,
\end{align*}
where
\begin{align*}
  a_+ := \sum_{i:a_i\geq 0} |a_i| \tr{\sigma_{i,+}}  +  \sum_{i:a_i<0} |a_i| \tr{\sigma_{i,-}}  \, \\
  a_- := \sum_{i:a_i\geq 0} |a_i| \tr{\sigma_{i,-}}  +  \sum_{i:a_i<0} |a_i| \tr{\sigma_{i,+}} 
\end{align*}
and 
\begin{align*}
  \sigma_+ := \frac{1}{a_+}\left(  \sum_{i:a_i\geq 0} |a_i|\sigma_{i,+}  +  \sum_{i:a_i<0} |a_i|\sigma_{i,-} \right) \in\SEP(A,B) \, , \\
  \sigma_- := \frac{1}{a_-}\left(  \sum_{i:a_i\geq 0} |a_i|\sigma_{i,-}  +  \sum_{i:a_i<0} |a_i|\sigma_{i,+} \right) \in\SEP(A,B) \, .
\end{align*}
This implies
\begin{align*}
  \gamma_{\SEP}(\rho_{AB}) \leq a_+ + a_- = \sum\limits_i |a_i| \left( \sigma_{i,+} + \sigma_{i,-}\right) \leq \sum\limits_i |a_i| = \gamma_{\LOCC}(\rho_{AB}) \, .
\end{align*}

For the second inequality, we consider the $a_{\pm}$ and $\rho_{\pm}$ which achieve the minimum in~\cref{eq_gamma_state_prep}.
Since the $\rho_{\pm}$ are separable, we can write them in the form $\rho_{+} = \sum_i p_{+,i} \sigma_{+,i}\otimes\tau_{+,i}$, $\rho_{-} = \sum_i p_{-,i} \sigma_{-,i}\otimes\tau_{-,i}$ for some probability distributions $p_{\pm,i}$ and some set of states $\sigma_{\pm,i},\tau_{\pm,i}$ on $A$ and $B$ respectively.
Take $\mathcal{F}_{+,i},\mathcal{F}_{-,j}\in \LO(A,B)$ to be the operations that prepare the state $\sigma_{+,i}\otimes\tau_{+,i}$ respectively $\sigma_{-,j}\otimes\tau_{-,j}$.
One has
\begin{align*}
    \gamma_{\LO}(\rho_{AB}) 
    &\leq \gamma_{\LO}(\sum_i a_{+}p_{+,i}\mathcal{F}_{+,i} - \sum_j a_{-}p_{-,j}\mathcal{F}_{-,j}) \\
    &\leq \sum_i a_+p_{+,i} + \sum_j a_-p_{-,j} \\
    &= a_++a_- \\
    &= \gamma_{\SEP}(\rho_{AB}) \, ,
\end{align*}
which proves the assertion. \qed
%%%%%%%%%%%%%%%%%%%%%%%%%%%%%%%%%%%%%%%%%%%%%%%%%%%%%%%%%%%%%%%%%%%%%%%%%%%%%%%%%%%%%%%%%%%%%%%%%%%%%%%%%
\subsection{Proof of~\cref{thm_main}} \label{app_proof_thm_main}
Since the $\gamma$-factor is invariant under local unitaries (see~\cref{lem_invariance_local_unitary}), we can without loss of generality assume that $U=\exp(\ci \theta_X X \otimes X +  \ci \theta_Y Y \otimes Y + \ci \theta_Z Z\otimes Z)$.
We show under the given assumptions that the lower bound from~\cref{lem_lower_bound} and the upper bound from~\cref{lem_upper_bound} coincide which then proves the assertion by~\cref{eq_trivial}.
Suppose first that Assumption~\ref{it_first} is fulfilled, i.e., $\theta_Z=0$. 
We then have
    \begin{align*}
\ket{\Phi_{U}}
\coloneqq&\sum_{i,j=0,1}\ket{ij}_{A'B'}\otimes U\ket{ij}_{AB}  \\  
=& \ee^{\ci\theta_X X_A X_B} \ee^{\ci\theta_Y Y_A Y_B} \frac{1}{2}\sum_{i,j}\ket{i}_{A'} \ket{i}_A \ket{j}_B \ket{j}_{B'} \\
=& \frac{1}{2}(\cos \theta_X \id_A \id_B\!+\!\ci\sin \theta_X X_A X_B)(\cos \theta_Y \id_A \id_B\!+\!\ci \sin\theta_Y Y_A Y_B)(\ket{0000} \!+\! \ket{0011} \!+\! \ket{1100} \!+\! \ket{1111}) \\
= & \frac{1}{2}\!\Big(\! (\!\cos \theta_X \cos \theta_Y\! +\! \sin \theta_X \sin \theta_Y\!)(\ket{0000} \!+\! \ket{1111}) \!+\! (\!\cos \theta_X \cos \theta_Y \!-\!\sin \theta_X \sin \theta_Y\!)(\ket{0011} \!+\! \ket{1100}) \\
& \hspace{-10mm}+\! (\ci \cos \theta_X \sin \theta_Y \!+\! \ci \sin \theta_X \cos \theta_Y\!)(\ket{0101} \!+\! \ket{1010}) \!+\! (\ci \sin \theta_X \cos \theta_Y \!-\! \ci \cos \theta_X \sin \theta_Y\!)(\ket{0110} \!+\! \ket{1001})\! \Big) \, ,
\end{align*}
where the the second step uses that $X_A X_B$ and $Y_A Y_B$ commute. The penultimate step follows from Sylvester's formula. 
To find the Schmidt coefficients we need to compute the singular values of following matrix
    \begin{align*}
      D= \frac{1}{2}
        \begin{pmatrix}
            \alpha_1 & 0 & 0 & \alpha_2 \\
            0 & \alpha_3 & \alpha_4 & 0 \\
            0 & \alpha_4 & \alpha_3 & 0 \\
            \alpha_2 & 0 & 0 & \alpha_1
        \end{pmatrix} \, .
    \end{align*}
for $\alpha_1:=\cos \theta_X \cos \theta_Y + \sin \theta_X \sin \theta_Y$, $\alpha_2:=\cos \theta_X \cos \theta_Y -\sin \theta_X \sin \theta_Y$, $\alpha_3:=\ci \cos \theta_X \sin \theta_Y + \ci \sin \theta_X \cos \theta_Y$, and $\alpha_4:=\ci \sin \theta_X \cos \theta_Y - \ci \cos \theta_X \sin \theta_Y$.   
The singular values and thus Schmidt coefficients of $D$ are $\{|\cos\theta_X \cos\theta_Z|$, $|\cos\theta_X\sin\theta_Z|$, $|\sin\theta_X\cos\theta_Y|$, $|\sin\theta_X\sin\theta_Y|\}$.
\cref{lem_lower_bound} thus gives
\begin{align}
  \gamma_{\LO}(U) 
  &\geq 2\big(|\cos\theta_X|+|\sin\theta_X|\big)^2 \big(|\cos\theta_Y|+|\sin\theta_Y|\big)^2-1 \nonumber \\
  & = 1 + 4|\sin\theta_X \cos \theta_X | + 4|\sin\theta_Y \cos \theta_Y | + 8|\sin\theta_X \cos\theta_X \sin\theta_Y \cos\theta_Y| \, . \label{eq_lowerbound}
\end{align}
    
We next show that the upper bound from~\cref{lem_upper_bound} coincides with~\cref{eq_lowerbound}.
To do so we first determine the coefficients $u_0,u_1,u_2,u_3$ according to~\cref{eq_standard_dec_u}. Therefore we write
    \begin{align*}
        &\exp(\ci\theta_X X_A X_B + \ci \theta_Y Y_A Y_B) \\
         &\hspace{5mm}= \exp(\ci\theta_X X_A X_B) \exp(\ci\theta_Y Y_A Y_B) \\
         &\hspace{5mm}= (\cos\theta_X \id_A \id_B+\ci\sin \theta_X X_A X_B)(\cos\theta_Y \id_A \id_B + \ci \sin\theta_Y Y_A Y_B) \\
         &\hspace{5mm}= \cos\theta_X \cos\theta_Y \id_A \id_B + \ci\sin\theta_X \cos\theta_Y X_A X_B +  \ci \cos\theta_X \sin\theta_Y Y_A Y_B - \sin\theta_X \sin\theta_Y Z_A Z_B \, .
    \end{align*}
So we have $u_0=\cos\theta_X \cos\theta_Y$, $u_1=\ci \sin\theta_X \cos\theta_Y$, $u_2= \ci \cos\theta_X \sin\theta_Y$ and $u_3=- \sin\theta_X \sin\theta_Y$.
Notice that all $u_i$ are only real or imaginary.
This implies that $|u_iu_j^* + u_ju_i^*| + |u_iu_j^* - u_ju_i^*| = 2|u_i||u_j|$.
    Therefore,~\cref{lem_upper_bound} gives
\begin{align}
\gamma_{\LO}(U)
&\leq 1  + 4|u_0||u_1| + 4|u_0||u_2| + 4|u_0||u_3| + 4|u_1||u_2| + 4|u_1||u_3| + 4|u_2||u_3| \nonumber \\
&=  1  + 4| \sin\theta_X \cos\theta_X | + 4 |\sin\theta_Y  \cos\theta_Y | + 8|\sin\theta_X \cos\theta_X \sin\theta_Y \cos\theta_Y | \, . \label{eq_upperBound}
    \end{align}
Combining~\cref{eq_lowerbound} with~\cref{eq_upperBound} proves the assertion.
%%%%%

Suppose next that Assumption~\ref{it_second} is fulfilled. In this case the gate is locally equivalent to a $\SWAP$ gate.
For the lower bound, this means that the Choi state is essentially $2$ Bell pairs and thus~\cref{corr_gamma_ebit} and~\cref{lem_lower_bound} imply
\begin{align} \label{eq_lowerbound2}
    \gamma_{\LOCC}(U) \geq 7 \, .  
\end{align}
For the upper bound, we first determine $u_0,u_1,u_2,u_3$ by writing
    \begin{align*}
          \exp \Big(\ci\frac{\pi}{4} X_A X_B + \ci\frac{\pi}{4} Y_A Y_B + \ci\frac{\pi}{4} Z_A Z_B\Big)
        = & \exp \Big( \ci\frac{\pi}{4} X_A X_B \Big) \exp\Big(\ci\frac{\pi}{4} Y_A Y_B \Big) \exp\Big(\ci\frac{\pi}{4} Z_A Z_B\Big) \\
        = & \frac{1}{2\sqrt{2}}(\id_A \id_B +\ci X_A X_B)(\id_A \id_B +\ci Y_A Y_B)(\id_A \id_B+\ci Z_A Z_B) \\
        = & \frac{1+\ci}{2\sqrt{2}}(\id_A \id_B+X_A X_B+Y_A Y_B+Z_A Z_B) \, ,
    \end{align*}
which implies $u_0=u_1=u_2=u_3=\frac{1+\ci}{2\sqrt{2}}$.
This means that $|u_iu_j^*+u_ju_i^*|+|u_iu_j^*-u_ju_i^*|=1/2$ and hence according to~\cref{lem_upper_bound}
\begin{align*}
    \gamma_{\LO}(U)  \leq  7 \, ,
\end{align*}
which together with~\cref{eq_lowerbound2} proves the assertion.
Note that
\begin{align*}
    1 + 4|\sin\theta_X \cos \theta_X | + 4|\sin\theta_Y \cos \theta_Y | + 8|\sin\theta_X \cos\theta_X \sin\theta_Y \cos\theta_Y| = 7
\end{align*}
for $\theta_X=\theta_Y=\theta_Z=\pi /4$.
\qed 
%%%%%%%%%%%%%%%%%%%%%%%%%%%%%%%%%%%%%%%%%%%%%%%%%%%%%%%%%%%%%%%%%%%%%%%%%%%%%%%%%%%%%%%%%%%%%%%%%%%%%%%%%%
\subsection{Proof of~\cref{cor_bounds_match}} \label{app_cor_gates}
We start with the $\mathrm{R}_{\sigma\sigma}(\theta)$ gate.
By using an appropriate local basis transformation and invoking~\cref{lem_invariance_local_unitary}, we can restrict our considerations to the $\mathrm{R}_{XX}(\theta)$ gate, which is iself equivalent to the $\mathrm{R}_{XX}(-\theta)$ gate up to local unitaries.
$\mathrm{R}_{XX}(-\theta)$ is already in the form of the KAK decomposition with $(\theta_X,\theta_Y,\theta_Z)=(\theta/2,0,0)$.
By~\cref{thm_main}, this directly implies
\begin{align*}
    \gamma_{\LOCC}\big(\mathrm{R}_{XX}(-\theta)\big) = \gamma_{\LOCCOW}\big(\mathrm{R}_{XX}(-\theta)\big) = \gamma_{\LO}\big(\mathrm{R}_{XX}(-\theta)\big) = 1+4 \Big|\sin\frac{\theta}{2}\cos\frac{\theta}{2}\Big| = 1 + 2|\sin\theta| \, .
\end{align*}

Next we consider the controlled-rotation gate $\mathrm{CR}_{\sigma}(\theta)$.
Again, by using an appropriate basis transformation on the target qubit, we can restrict our considerations to the $\mathrm{CR}_{X}(\theta)$ gate.
It is equivalent to the $\mathrm{R}_{XX}(-\theta/2)$ gate up to local unitaries:
\begin{align}\label{eq_equivalence_crx_rxx}
    \mathrm{CR}_{X}(\theta) = \big(H\otimes \mathrm{R}_X(\theta /2) \big) \mathrm{R}_{XX}(-\theta /2) (H\otimes \id) \, ,
\end{align}
where $H$ is the Hadamard gate and $\id$ the identity gate.
Hence, we have
\begin{align*}
    \gamma_{\LOCC}\big(\mathrm{CR}_{X}(\theta)\big) = \gamma_{\LOCCOW}\big(\mathrm{CR}_{X}(\theta)\big) = \gamma_{\LO}\big(\mathrm{CR}_{X}(\theta)\big)= 1 + 2|\sin(\theta /2)| \, .
\end{align*}

As mentioned in the main text, every two-qubit Clifford gate is equivalent (up to local unitaries) to either the identity gate $\identity$, the $\CNOT$ gate, the $\iSWAP$ gate or the $\SWAP$ gate.
By~\cref{lem_invariance_local_unitary}, it therefore suffices to only consider these four gates.
The identity gate trivially fulfills $\theta_X=\theta_Y=\theta_Z=0$.
The $\CNOT$ gate has a KAK decomposition with angles $(\theta_X=\pi/4,\theta_Y=0,\theta_Z=0)$, the $\iSWAP$ gate has $(\theta_X=\pi/4,\theta_Y=\pi/4,\theta_Z=0)$, and the $\SWAP$ gate has $(\theta_X=\pi/4,\theta_Y=\pi/4,\theta_Z=\pi/4)$~\cite{HK21}. As a result,~\cref{thm_main} implies
\begin{alignat*}{4}
        &\gamma_{\LOCC}(\CNOT)  &&= \gamma_{\LOCCOW}(\CNOT)  &&= \gamma_{\LO}(\CNOT)  &&= 3 \\
      &\gamma_{\LOCC}(\iSWAP)  &&= \gamma_{\LOCCOW}(\iSWAP)  &&= \gamma_{\LO}(\iSWAP)  &&= 7 \\
      &\gamma_{\LOCC}(\SWAP)  &&= \gamma_{\LOCCOW}(\SWAP)  &&= \gamma_{\LO}(\SWAP)  &&= 7 \, ,  
\end{alignat*}
which completes the proof. \qed

%%%%%%%%%%%%%%%%%%%%%%%%%%%%%%%%%%%%%%%%%%%%%%%%%%%%%%%%%%%%%%%%%%%%%%%%%%%%%%%%%%%%%%%%%%%%%%%%%%%%%%%%%%
%%%%%%%%%%%%%%%%%%%%%%%%%%%%%%%%%%%%%%%%%%%%%%%%%%%%%%%%%%%%%%%%%%%%%%%%%%%%%%%%%%%%%%%%%%%%%%%%%%%%%%%%%%
\bibliographystyle{arxiv_no_month}
\bibliography{bibliofile}

\end{document}